\documentclass[journal,fleqn]{IEEEtran}
\usepackage{amsmath} 
\usepackage{amsthm}
\usepackage{soul}
\usepackage{graphicx}
\usepackage{cases}
\usepackage{epstopdf} 
\graphicspath{{Figures/}}
\usepackage{algorithm}
\usepackage{amsfonts}
\usepackage{algcompatible}
\usepackage{xpatch}  
\usepackage{caption}
\usepackage{mathabx}
\usepackage{mathtools}
\usepackage[table]{xcolor}
\usepackage[utf8]{inputenc}
\usepackage{textcomp} 
\usepackage{nomencl}
\usepackage{array}
\usepackage{multirow}
\usepackage{float}
\usepackage[utf8]{inputenc}
\usepackage{soul} 
\newcommand{\comment}[1]{}
\newcommand{\w}{\!\!\!\!\!}
\makenomenclature
\usepackage[noadjust]{cite}
 
\newtheorem{theorem}{Theorem}

\newtheorem{remark}{Remark}

\makeatletter
\newcommand{\mathleft}{\@fleqntrue\@mathmargin0pt}
\newcommand{\mathcenter}{\@fleqnfalse}
\xpatchcmd{\algorithmic}
{\ALG@tlm\z@}{\leftmargin\z@\ALG@tlm\z@}
{}{}
\makeatother

\begin{document}
	\title{Power Minimization in Distributed Antenna Systems using Non-Orthogonal Multiple Access and Mutual Successive Interference Cancellation}
	
	\author{~Joumana~Farah,~Antoine Kilzi,~Charbel~Abdel~Nour,~Catherine~Douillard
		\thanks{Copyright (c) 2015 IEEE. Personal use of this material is permitted. However, permission to use this material for any other purposes must be obtained from the IEEE by sending a request to pubs-permissions@ieee.org.}
		\thanks{J. Farah is with the Department of Electricity and Electronics,
			Faculty of Engineering, Lebanese University, Roumieh, Lebanon
			(\mbox{joumanafarah}@ul.edu.lb).}
		\thanks{A. Kilzi, C. Abdel Nour and C. Douillard are with IMT Atlantique, Lab-STICC, UBL, F-29238 Brest, France, (email: antoine.kilzi@imt-atlantique.fr; charbel.abdelnour@imt-
			atlantique.fr; catherine.douillard@imt-atlantique.fr).}}
	
	\maketitle
	\begin{abstract}
		This paper introduces new approaches for combining non-orthogonal multiple access with distributed antenna systems. The study targets a minimization of the total transmit power in each cell, under user rate and power multiplexing constraints. Several new suboptimal power allocation techniques are proposed. They are shown to yield very close performance to an optimal power allocation scheme. Also, a new approach based on mutual successive interference cancellation of paired users is proposed. Different	techniques are designed for the joint allocation of subcarriers, antennas, and power, with a particular care given to maintain a moderate complexity. The coupling of non-orthogonal multiple access to distributed antenna systems is shown to greatly outperform any other combination of orthogonal/non-orthogonal multiple access schemes with distributed or centralized deployment scenarios.
		
	\end{abstract}
	
	\section{Introduction}
	\IEEEPARstart{T}{he}  concept of distributed antenna systems (DAS), also known as distributed base stations, \cite{R1,R2} was introduced in the past few years in mobile communication systems to increase the cell coverage in a cost effective way, and to strengthen the network infrastructure, particularly in saturated areas. It consists of deploying the base station (BS) antennas in a distributed manner throughout the cell, instead of having multiple antennas installed on a single tower at the cell center. The remote units, called remote radio heads (RRH), are connected to the baseband unit (BBU) through high capacity coaxial cables or fiber optics. By reducing	the average distance of each mobile user to its transmitting/receiving antenna, the overall transmission power required to guarantee a certain quality of reception is reduced in comparison to the centralized configuration (centralized antenna system or CAS). Therefore, from an ecological standpoint, DAS can greatly reduce local electromagnetic radiation and $\text{CO}_2$ emissions of transmission systems. Alternatively, for the same overall transmission power as in CAS, DAS offers a higher capacity and a fairer throughput distribution between the active users of a cell. Moreover, it provides a	better framework for improving the system robustness to fading,	intra-cell and inter-cell interferences, shadowing, and path-loss. It also allows the system to better adapt to the varying user distribution. Moreover, the decoupled architecture will allow the deployment of small antennas in large scale and in discrete locations in urban areas, e.g. on building roofs, electric poles, traffic and street lights, where they can be almost invisible due to their small size. This will significantly simplify and reduce the cost of site installation, therefore lowering the capital expenditure (CAPEX) of mobile operators.
	
	Efficient implementation is key in squeezing the achievable	potentials out of DAS. For this purpose, the study in \cite{R3}
	explored the advantages of DAS and compared the achievable ergodic capacity for two different transmission scenarios: selection diversity and blanket transmission. In the first one, one of the RRHs is selected (based on a path-loss minimization criterion) for transmitting a given signal, whereas in the second, all antennas in the cell participate in each transmission, thus creating a macroscopic multiple antenna system. The results of \cite{R3} show that selection diversity achieves a better capacity in the DAS context, compared to blanket transmission. The same observations are made in \cite{R4}. In \cite{R5}, RRH selection is also preconized as a mean to decrease the number of information streams that need to be assembled from or conveyed to the	involved RRHs, as well as the signaling overhead. 
	\subsection{Energy Efficiency Maximization in DAS}
	Several works target the optimization of system energy efficiency (EE) in the DAS context. In \cite{R6}, two antenna selection techniques are proposed, either based on user path-loss information or on RRH energy consumption. Also, proportional fairness scheduling is considered for subband allocation with a utility function adapted to optimize the EE. In \cite{R7}, subcarrier assignment and power allocation (PA) are done in two separate stages. In the first one, the number of subcarriers per RRH is determined, and subcarrier/RRH assignment is performed assuming initial equal power distribution. In the second stage, power allocation (PA) is performed by maximizing the EE under the constraints of the total transmit power per RRH, of the targeted bit error rate and of a proportionally-fair throughput distribution among active users. The optimization techniques proposed in \cite{R6,R7} for DAS are designed for the case of orthogonal multiple access (OMA). In other words, they allow the allocation of only one user per subcarrier.
	\subsection{Power Domain Non-Orthogonal Multiple Access}
	Non-orthogonal multiple access (NOMA) has recently emerged as a promising multiple access technique to significantly improve the attainable spectral efficiency for fifth generation (5G) mobile networks. Power-domain NOMA enables the access of multiple users to the same frequency resource at the same time by taking advantage of the channel gain difference between users \cite{R9,R10,R11,R12,R13}, through signal power multiplexing. At the receiver side, user separation is performed using successive interference cancellation (SIC). Applying power multiplexing on top of the orthogonal frequency division multiplexing (OFDM) layer has proven to significantly increase system throughput compared to orthogonal signaling, while also improving fairness and cell-edge user experience. A few previous works have studied the application of NOMA in the DAS context. An outage probability analysis for the case of two users in cloud radio access networks (C-RAN) is provided in \cite{CRANZIGO} where all RRHs serve simultaneously both users. The results show the superiority of NOMA when compared to time division multiple access (TDMA), in the context of C-RANs. In \cite{CRANUplinkCooperation}, the study investigates the application of distributed NOMA for the uplink of C-RANs. The partially centralized C-RAN architecture allows the use of joint processing by distributed antennas, in which RRHs can exchange correctly decoded messages from other RRHs in order to perform SIC. In \cite{FronthaulCRAN}, an efficient end-to-end uplink transmission scheme is proposed where the wireless link between users and RRHs on one side, and the fronthaul links between the RRHs and BBU on the other side are studied. User grouping on blocks of subcarriers is proposed to mitigate the computational complexity, and a fronthaul adaptation for every user group is performed in order to strike a tradeoff between throughput and fronthaul usage.
	
	\subsection{State of the Art of Power Minimization in the NOMA Context}
	A few recent works tackle the power minimization problem in the NOMA context. In \cite{R14}, a ``relax-then-adjust'' procedure is used to provide a suboptimal solution to the NP-hard problem: first, the problem is relaxed from the constraints relative to power domain multiplexing. Then, the obtained solution is iteratively adjusted using a bisection search, leading to a relatively high complexity. In \cite{R15}, optimal PA is first conducted assuming a predefined fixed subcarrier assignment. Then, a deletion-based algorithm iteratively removes users from subcarriers until the constraints of the maximum number of	multiplexed users are satisfied, thus necessitating a large number of iterations to converge. In \cite{R16}, the authors propose an optimal and a suboptimal solution for determining the user scheduling, the SIC order, and the PA, for the case of a maximum of two users per subcarrier. However, the power multiplexing constraints are not taken into consideration. \textcolor{black}{The power multiplexing constraints state that the signal that is to be decoded first must have a higher power level than the other received signals, so that it is detectable at the receiver side.} Power minimization	strategies are also proposed in \cite{R17} for multiple-input multiple-output NOMA (MIMO-NOMA), where PA and receive beamforming design are alternated in an iterative way. Constraints on the targeted SINR (signal to interference and noise ratio) are considered to guarantee successful SIC decoding. The subcarrier allocation problem is not included, i.e. all users have access to the whole spectrum. Results, provided for a moderate number of users (4 or 6), show an important gain in performance with respect to OMA.

	In \cite{R18}, we have introduced a set of techniques that allow the joint allocation of subcarriers and power, with the aim of minimizing the total power in NOMA-CAS. Particularly, we showed that the most efficient method, from the power perspective, consists of applying user pairing at a subsequent stage to single-user assignment, i.e. after applying OMA signaling at the first stage, instead of jointly assigning collocated users to subcarriers. To the best of our knowledge, no previous work has studied the downlink power minimization problem in a DAS configuration and using NOMA.
	\subsection{Contributions}
	The main objective of this work is to study the potentials of applying NOMA in the DAS configuration from a power minimization perspective. We investigate the resource allocation (RA) problem in downlink, seeking the minimization of the total transmit power at the RRHs under user rate constraints. The contributions of this paper are summarized in the following :
	\begin{itemize}
		\item 
		We introduce several techniques that allow a significant complexity reduction of the waterfilling procedures used for PA in \cite{R18}, for both orthogonal and non-orthogonal transmission, while adapting the allocation techniques to the DAS context.
		\item We propose a new PA scheme for user pairing that outperforms FTPA (fractional transmit power allocation) \cite{R10,R11}, while taking into account the power multiplexing constraints.
		\item Unlike previous works, we investigate the use of different RRHs to power the multiplexed subcarriers in NOMA. This new setting gives rise to the concept of ``mutual SIC'' where paired users on a subcarrier can perform SIC at the same time, under well defined conditions.
		\item Finally, we propose new suboptimal algorithms to achieve joint subcarrier, RRH, and power allocation, in light of the newly uncovered potentials specific to the application of NOMA in the DAS context.
	\end{itemize}
	
	The rest of this paper is organized as follows: in Section II, a description of the system model is provided along with a formulation of the RA problem in the context of NOMA-DAS. Then, in Section III, several suboptimal solutions are investigated for the power minimization problem in the case of a single powering RRH per subcarrier. In Section IV, a novel approach allowing a mutual SIC implementation on certain subcarriers is introduced, followed by the proposal of several allocation techniques for exploiting such subcarriers. Section V provides a brief overview of the complexity of the proposed algorithms. Section VI presents a performance analysis of the different allocation strategies, while Section VII concludes the paper.	
	\section{Description of the NOMA-DAS System and Formulation of the Power Minimization Problem}
	This study is conducted on a downlink system consisting of a total of $R$ RRHs uniformly positioned over a cell where $K$ mobile users are randomly deployed (Fig. \ref{fig:1}). The RRHs are connected to the BBU through high capacity optical fibers. RRHs and users are assumed to be equipped with a single antenna. Users transmit their channel state information (CSI) to RRHs, and the BBU collects all the CSI from RRHs. The influence of imperfect channel estimation on the performance of DAS was studied in \cite{New}. However, perfect CSI is assumed throughout this study (the influence of imperfect or outdated CSI is not the aim of this work). Alternatively, the BBU can benefit from channel reciprocity to perform the downlink channel estimation by exploiting the uplink transmissions. Based on these estimations, the BBU allocates subcarriers, powers, and RRHs to users in such a way to guarantee a transmission rate of $R_{k,req}$ [bps] for each user $k$. The system bandwidth $B$ is equally divided into $S$ subcarriers. Each user $k$ is allocated a set $\mathcal{S}_k$ of subcarriers. From the set of $K$ users, a maximum of $m(n)$ users $\{k_1,k_2,\dots,k_{m(n)}\}$ are chosen to be collocated on the $n^{th}$ subcarrier $(1\leq n \leq S)$. Classical OMA signaling corresponds to the special case of $m(n) = 1$. Let	$h_{k_i,n,r}$ be the channel coefficient between user $k_i$ and RRH $r$ over subcarrier $n$, and $H$ the three-dimensional channel gain matrix with elements $h_{k,n,r}$, $1 \leq k \leq K$, $1 \leq n \leq S$, $1\leq r \leq R$. A user $k_i$ on subcarrier $n$ can remove the inter-user interference from any other user $k_j$, collocated on subcarrier $n$, whose channel gain	verifies $h_{k_j,n} < h_{k i , n}$ \cite{R9,R10} and treats the received signals from	other users as noise.
	
	As shown in Fig. \ref{fig:1}, NOMA subcarriers can be served by the same RRH or by different RRHs. For instance, one can consider serving User 1 and User 2 on the same subcarrier SC 1 by RRH 1, while User 2 and User 3 are paired on another subcarrier, SC 2, and served by RRH 1 and RRH 2 respectively.
	\begin{figure}[H]
		\centering
		\includegraphics[width=0.4\textwidth]{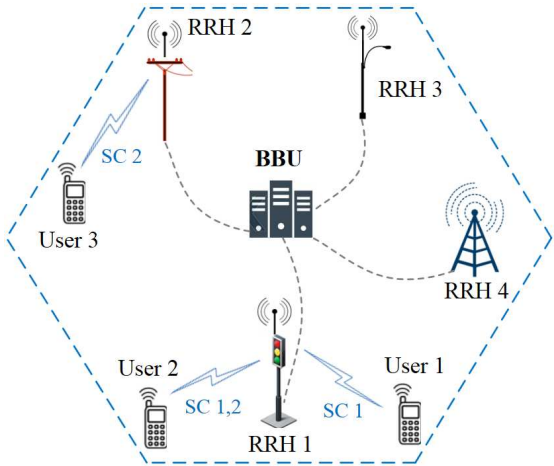}
		\captionsetup{justification=centering}
		\caption[centered]{DAS using NOMA}
		\label{fig:1}
	\end{figure}
	The capacity of the BBU-RRH links is assumed to be many orders of magnitude higher than the capacity of the RRH-users wireless links. This is due to the radical differences between the two channel media: interference immune, low loss, high bandwidth fibers with dedicated channels as opposed to the frequency selective, time varying, shared medium that is the wireless link. Clearly, the bottleneck of the system, in terms of capacity and power consumption, resides at the wireless link level, where large margins of improvement can be achieved.

	In the rest of the study, and without loss of generality, we will consider a maximum number of collocated users per subcarrier of 2, i.e. $m(n)$= 1 or 2. On the one hand, it has been shown \cite{R10} that the gain in performance obtained with the collocation of 3	users per subcarrier, compared to 2, is minor. On the other hand, limiting the number of multiplexed users per subcarrier limits the SIC complexity at the receiver terminals. We will denote by	first (resp. second) user on a subcarrier $n$ the user which has the higher (resp. lower) channel gain on $n$ between the two paired users. Let $P_{k_i,n,r}$ be the power of the $i^{th}$ user on subcarrier $n$ transmitted by RRH $r$. The theoretical throughputs $R_{k_i , n , r} , 1 \leq i \leq 2$, on $n$ are given by the Shannon capacity limit as follows:
	\begin{equation}
	R_{k_1,n,r} = \frac{B}{S} \log_2\left(1+\frac{P_{k_1,n,r}h_{k_1,n,r}^2}{\sigma^2}\right),
	\end{equation}
	\begin{equation}
	\label{rate2equation}
	R_{k_2,n,r} = \frac{B}{S} \log_2\left(1+\frac{P_{k_2,n,r}h_{k_2,n,r}^2}{P_{k_1,n,r}h_{k_2,n,r}^2+\sigma^2}\right), 
	\end{equation}
	where $N_0$ and $\sigma^2=N_0B/S$ are respectively the power spectral density and the power level (over a subcarrier) of additive white Gaussian noise,
	including randomized inter-cell interference, and assumed to be
	constant over all subcarriers.\\

	Let $\mathcal{T}_k$ be the	mapping set of RRHs corresponding to user $k$, such that the $i^{th}$	element of $\mathcal{T}_k$ corresponds to the RRH selected for powering the $i^{th}$ subcarrier from $\mathcal{S}_k$. Note that user $k$ can be first, second, or sole user on any of its allocated subcarriers in $\mathcal{S}_k$. Taking into account the power multiplexing constraints proper to NOMA, the corresponding optimization problem can be formulated as:
	\begin{equation*}
	\{\mathcal{S}_k,\mathcal{T}_k,P_{k,n,r}\}^* = \mathop{\arg\min}_{\{\mathcal{S}_k,\mathcal{T}_k,P_{k,n,r}\}}
	\sum_{k=1}^{K}\qquad\sum_{\mathclap{\substack{n \in \mathcal{S}_k\\r=\mathcal{T}_k(i),\text{ s.t. }\mathcal{S}_k(i)=n}}}P_{k,n,r},
	\end{equation*}
	Subject to:
	\begin{numcases}{}
	\sum_{n \in \mathcal{S}_k}R_{k,n,r} = R_{k,req}, \forall k, 1 \leq k \leq K \label{condition 1}\\
	P_{k,n,r} \geq 0, n \in \mathcal{S}_k, r \in \mathcal{T}_k, 1 \leq k \leq K \label{condition 2}\\
	P_{k_2,n,r} \geq P_{k_1,n,r}, \forall n \in \mathcal{S}_k , 1 \leq k \leq K \label{condition 3}
	\end{numcases}
	\textcolor{black}{The problem consists in finding the optimal subcarrier-RRH-user allocation, as well as the optimal power allocation over the allocated subcarriers, so as to minimize the objective function that is the total transmit power of the cell. This must be done under the rate constraints (\ref{condition 1}), positive power constraints (\ref{condition 2}), and power multiplexing constraints (\ref{condition 3}). The first constraint imposes a minimum rate requirement $R_{k,req}$ for every user $k$,  that must be achieved over the subcarriers $\mathcal{S}_k$ allocated to $k$. The second condition ensures that all power variables remain non-negative (a null power variable corresponds to an unallocated subcarrier). Finally, the last constraint accounts for the power multiplexing conditions where the power level of the signal of the weak user, $P_{k_2,n,r}$, must be greater than the power level of the signal of the strong user, $P_{k_1,n,r}$.} \textcolor{black}{Solving this optimization problem resides in determining the optimal allocation sets $\mathcal{S}_k$, $\mathcal{T}_k$ for every user $k$, as well as finding the optimal power allocation over the allocated subcarriers. Therefore, the optimization problem at hand is mixed combinatorial and non-convex, that is why we resort to suboptimal solutions for the joint subcarrier-RRH-user assignment and power allocation problem.} 
	Moreover, compared to the case of NOMA-CAS signaling, an additional dimension is added, corresponding to the determination of the best RRH to power each multiplexed subcarrier. Therefore, we study separately the cases of powering the multiplexed signals on a subcarrier by a common RRH (Section III), or by different RRHs (Section IV). We optimize the performance of both allocation schemes and then combine the RRH selection strategies into a unified algorithm which outperforms its predecessors.
	
	\section{Resource Allocation Techniques for the Case a Single Powering RRH per Subcarrier }
		
	\subsection{Power minimization procedure }\label{sec:RuntimeEnhancment}	
Given the intractability of the optimal subcarrier and power allocation solutions, a greedy approach was introduced in \cite{R18} for the CAS context, which aims at minimizing the power decrease resulting from every new subcarrier allocation: user-subcarrier assignment is determined by first selecting the most power consuming user and then allocating to it the available subcarrier that best reduces its requested transmit power from the base station. \textcolor{black}{Note that in DAS, this step resides in the selection of the best subcarrier-RRH couple from the space of $S\times R$ subcarrier-antenna pairs}.

In the OMA phase, power allocation is performed using a recursive low-complexity waterfilling technique. For this reason, we start by revisiting the waterfilling principle in order to introduce several
procedures for reducing the complexity of both orthogonal and non-orthogonal RA phases. Then, we extend the obtained solution to the DAS context.\\
	
	Given a user $k$ allocated a total of $N_k$ subcarriers and having a waterlevel $w_k(N_k)$, the addition of a subcarrier-RRH pair $(n_a,r)$ in the OMA phase decreases the waterline if and only if its channel gain verifies \cite{R18}:
	\begin{equation}\label{Omachannelcondition}
	h^2_{k,n_a,r} > \frac{\sigma^2}{w_k(N_k)}
	\end{equation}The new waterlevel, as well as the power incurred by the subcarrier assignment in the OMA phase, can be determined by means of the following equations:
	\begin{align}
	&w_k (N_k+1)=\frac{(w_k (N_k))^{N_k/(N_k+1)}}{(h_{k,n_a,r}^2/\sigma^2)^{1/(N_k+1)}}\label{waterdrop}\\
	&\Delta P_{k,n_a,r} = (N_k+1)w_k(N_k+1)-N_kw_k(N_k)-\frac{\sigma^2}{h^2_{k,n_a,r}}\label{DeltaPOMA}
	\end{align}
	For more detailed information about (\ref{Omachannelcondition}), (\ref{waterdrop}), and (\ref{DeltaPOMA}), the reader is referred to \cite{R18}. Note that when a CAS configuration is considered, $r$ designates the central (unique) BS antenna.The subcarrier minimizing the power decrease  for user $k$ is also the subcarrier with the best channel for this user, as shown in Appendix \ref{app:1}. This equivalence will allow us to decrease the complexity of the OMA phase with respect to \cite{R18}.
	
	Next, the user pairing  phase is considered, i.e. the assignment of second users to subcarriers in the NOMA phase. As stated in \cite{R18}, when allocating a subcarrier to user $k$ as second user, the waterline of the solely occupied subcarriers by $k$ must be decreased in order to avoid any excess in rate (compared to its required rate) and thus in power. In addition, the initial waterlevel for every user $k$ in the NOMA phase is the final waterline obtained in the OMA phase. In \cite{R18}, a dichotomy-based waterfilling technique \cite{R19} is used after each new pairing in the NOMA phase to determine the power level on each sole subcarrier of user $k$. For this purpose, we derive next an alternative iterative waterfilling calculation for the NOMA phase, with a significant complexity reduction compared to dichotomy-based waterfilling. 
	
	Let  $\mathcal{S}_k^{sole}$ be the set of solely allocated subcarriers to user $k$, $N_k^{sole}$ the cardinal of this set denoted by $N_k^{sole} = \vert \mathcal{S}_k^{sole}\vert$, and $R_k^{sole}$ the total rate achieved by user $k$ over its subcarriers in $\mathcal{S}_k^{sole}$. We have:
	\begin{equation*}
	R_k^{sole} = \sum_{\mathclap{\substack{n \in \mathcal{S}_k\\r=\mathcal{T}_k(i),\text{ s.t. }\mathcal{S}_k(i)=n}}}\frac{B}{S} \log_2\Bigg(1+\frac{P_{k,n,r}h_{k,n,r}^2}{\sigma^2}\Bigg),
	\end{equation*}
	with $P_{k,n,r} = w_k(N_k^{sole})-\sigma^2/h_{k,n,r}^2$, where $w_k(N_k^{sole})$ is the waterline corresponding to $\mathcal{S}_k^{sole}$. \\ Therefore, $R_k^{sole}$ can be rewritten as:
	\begin{align}
	&R_k^{sole}=\sum_{n \in \mathcal{S}_k^{sole}}\frac{B}{S} \log_2\bigg(\frac{w_k(N_k^{sole})h_{k,n,r}^2}{\sigma^2}\bigg) \text{, leading to:} \nonumber\\
	&w_k(N_k^{sole}) = \bigg(2^{R_k^{sole}S/B}\prod_{n \in \mathcal{S}_k^{sole}}\frac{\sigma^2}{h_{k,n,r}^2}\bigg)^{\frac{1}{N_k^{sole}}} \label{waterleveldrop}
	\end{align}
	On the other hand, $R_k^{sole}$ is calculated by subtracting from $R_{k,req}$ the rates of user $k$ on the subcarriers where $k$ is either first or second user. When user $k$ is assigned as a second user to a subcarrier $n$, the corresponding rate gain is calculated using (\ref{rate2equation}), with a power level on $n$ given by FTPA:
	\[
	P_{k_2,n,r}=P_{k_1,n,r}h_{k_2,n,r}^{-2\alpha}/h_{k_1,n,r}^{-2\alpha}
	\]
	where $P_{k_1,n,r}$ is the power of the first user previously allocated to subcarrier $n$ in the OMA phase.
	This additional rate corresponds to the rate decrease $\Delta R_{k,n,r}$ that should be compensated for on the sole subcarriers of $k$, so as to ensure the
	global rate constraint $R_{k,req}$. In other words, the variation that the rate $R_k^{sole}$ undergoes is opposite to the rate addition that comes along the new subcarrier assignment. We can write the new rate that must be achieved on $\mathcal{S}_k^{sole}$ as $R_k^{sole'} = R_k^{sole}+\Delta R_{k,n,r}$ where the rate decrease $\Delta R_{k,n,r}$ is negative. Using (\ref{waterleveldrop}), the new waterline on the set $\mathcal{S}^{sole}_k$ can then be derived as follows:
	\begin{align}
	&w_k'(N_k^{sole})=\Bigg(2^{R_k^{sole'}S/B}\prod_{n \in \mathcal{S}_k^{sole}}\frac{\sigma^2}{h_{k,n,r}^2}\Bigg)^{\frac{1}{N_k^{sole}}} \nonumber \\
	&w_k'(N_k^{sole})=2^{\frac{\Delta R_{k,n,r}S}{BN_k^{sole}}}w_k(N_k^{sole})\label{equ:NOdichotomy}
	\end{align}
	This expression of the waterline update in the NOMA phase enables the fast computation  of the potential power decrease corresponding to any candidate subcarrier-RRH pair, using:
	\begin{align}
	\Delta P_{k_2,n,r} &= N_{k_2}^{sole}\bigg(w_{k_2}^{\prime}(N_{k_2}^{sole})-w_{k_2}(N_{k_2}^{sole}\bigg) + P_{k_2,n,r}\label{DeltaPNOMA}
	\end{align}
	This reduced complexity algorithm can then be directly extended to the DAS context. \textcolor{black}{In the OMA phase, the subcarrier selection in CAS turns into a subcarrier-RRH pair assignment in DAS}. Therefore, the space of possible links to attribute to the user is enlarged by the factor $R$. However, in the NOMA phase, we restrict the transmission of the paired subcarrier to the same RRH serving the first user. Indeed, when different RRHs are chosen to power the multiplexed subcarrier, special features need to be addressed as will be shown in Section \ref{InformationTheory}. This NOMA-DAS method will be referred to as SRRH (meaning Single RRH per subcarrier). The corresponding details are presented in Algorithm \ref{alg:NOM-SRRH}, where $\mathcal{U}_p$ is the set of users whose power level can still be decreased, $\mathcal{S}_p$ is the set of unallocated subcarriers, and $\mathcal{S}_f$ is the set of subcarriers assigned a first user without a second user.
	\begin{algorithm}[H]
		\caption{SRRH}\label{alg:NOM-SRRH}
		\begin{algorithmic}
			\STATE\textbf{\small Initialization:} \text{ \small $\mathcal{S}_p= \textlbrackdbl 1:S \textrbrackdbl$}
			\STATE\text{\small \qquad\qquad\qquad $\mathcal{U}_p = \textlbrackdbl 1:K \textrbrackdbl$}
			\STATE\text{\small\qquad\qquad\qquad $\mathcal{S}_f = \emptyset $}
			\STATE\STATE\textbf{\small Phase 1:} \textcolor{gray}{\small// Worst-Best-H phase}
			\STATE\small Take the user whose best subcarrier-RRH link is the lowest among users and assign it its best subcarrier-RRH pair with the needed power to reach $R_{k,req}$. Repeat until all users have one allocated subcarrier-RRH pair moved from $\mathcal{S}_p$ to $\mathcal{S}_f$.
			\STATE \STATE 
			\textbf{\small Phase 2: }\textcolor{gray}{\small// Orthogonal multiplexing (single-user assignment)} 
			\STATE $k^*\!=\!\mathop{\arg\max}\limits_k P_{k,tot} \text{\small \textcolor{gray}{//identify the most power-consuming user}}$
			\STATE \text{$(n^*,r^*) = \mathop{\arg\max}\limits_{(n,r),\text{ s.t. }n\in \mathcal{S}_p \& (\ref{Omachannelcondition})} h_{k^*,n,r}$ \textcolor{gray}{// identify its most favorable}}
			\STATE \text{\qquad\qquad\qquad\qquad\quad\qquad\quad\quad\quad\;\; \textcolor{gray}{// subcarrier-RRH pair}}
			\STATE Calculate $w_{k^*}(N_{k^*}+1),\Delta P_{k^*,n^*,r^*}$ using (\ref{waterdrop}) and (\ref{DeltaPOMA})
			\STATE \textbf{\small If} $\Delta P_{k^*,n^*,r^*}<-\rho$ \textcolor{gray}{\small// $(n^*,r^*)$ allows a significant power}
			\STATE \text{\qquad\qquad\qquad\qquad\quad\quad\;\, \textcolor{gray}{// decrease}}
			\STATE \text{\small\quad Attribute $(n^*,r^*)$ to $k^*$,}
			\STATE \text{\small\quad Remove $n^*$ from $\mathcal{S}_p$,}
			\STATE \text{\small\quad Add $n^*$ to $\mathcal{S}_f$,}
			\STATE \text{\small\quad Update $P_{k^*,tot}$}
			\STATE \textbf{\small Else }remove $k^*$ from $\mathcal{U}_p$ \textcolor{gray}{\small // $k^*$'s power can no longer be decreased}
			\STATE \text{\qquad\qquad\quad\qquad\quad\quad\quad\;  \textcolor{gray}{// in OMA}}
			\STATE\STATE Repeat Phase 2 until no more subcarriers can be allocated
					\STATE	\STATE \text{\small$\mathcal{U}_p=\textlbrackdbl1:K\textrbrackdbl$}
			\algstore{Pause1}
		\end{algorithmic}
	\end{algorithm}
	\begin{algorithm}
		\begin{algorithmic}
			\algrestore{Pause1}
			\STATE \textbf{\small Phase 3: } \textcolor{gray}{\small // NOMA pairing}
			\STATE
			\STATE $k_2=\mathop{\arg\max}\limits_k P_{k,tot}$
			\STATE \textbf{\small For }every $n \in \mathcal{S}_f$ s.t. $h_{k_2,n,r}<h_{k_1,n,r}$ \textcolor{gray}{\small // $r$ is the RRH powering}
			\STATE \text{\qquad\qquad\qquad\qquad\qquad\qquad\qquad\qquad\; \textcolor{gray}{// user $k_1$ on $n$}}
			\STATE \text{\small\quad Calculate $P_{k_2,n,r}$ through FTPA}
			\STATE \text{\small\quad Calculate $w^{\prime}_{k_2}(N_{k_2}^{sole})$ using (\ref{equ:NOdichotomy})}
			\STATE \text{\small\quad Calculate $\Delta P_{k_2,n,r}$} using (\ref{DeltaPNOMA})
			\STATE \textbf{\small End for}
			\STATE $n^*=\mathop{\arg\min}\limits_n \Delta P_{k_2,n,r}$
			\STATE \textbf{\small If }$\Delta P_{k_2,n^*,r}<-\rho$
			\STATE \text{\small \quad Assign $k_2$ on $n^*$ and remove $n^*$ from $\mathcal{S}_f$}
			\STATE \text{\small \quad Fix $P_{k_1,n^*,r^*}$ and $P_{k_2,n^*,r^*}$, update $P_{k_2,n,r}, \forall n\in \mathcal{S}_{k_2}^{sole}$}
			\STATE \textbf{\small Else} \text{remove $k_2$ from $\mathcal{U}_p$}
			\STATE \small Repeat Phase 3 until $\mathcal{S}_f = \emptyset \lor \mathcal{U}_p=\emptyset$
		\end{algorithmic}
	\end{algorithm}
	
	In Algorithm \ref{alg:NOM-SRRH}, we start by ensuring that all users reach their targeted rates in the Worst-Best-H phase. From that point onward, the total power of the system decreases with every subcarrier allocation (by at least $\rho$).
	
	The threshold $\rho$ is chosen in such a way to strike a balance between the power efficiency and the spectral efficiency of the system, since unused subcarriers are released for use by other users or systems. Each time a user $k_2$ is paired with a user $k_1$ on a subcarrier $n^*$, the powers of $k_1$ and $k_2$ on $n^*$ are kept unvaried, i.e. they will no longer be updated at subsequent iterations. Actually, in both phases 2 and 3, an iteration results in either the allocation of a subcarrier-RRH pair, or the dismissing of a user from the set $\mathcal{U}_p$ of active users (in case of a negligible power decrease). Either ways, the total number of available subcarriers or active users is decreased by one in every iteration. Therefore, phases 2 and 3 involve at most $\vert\mathcal{S}_p\vert-K$ and $\vert\mathcal{S}_f\vert$ iterations respectively. In the worst case scenario, we have $\vert\mathcal{S}_f\vert =\vert\mathcal{S}_p\vert =S$. These considerations are central to the complexity analysis led in Section \ref{Sec:Complexity}, and they prove the stability of Algorithm \ref{alg:NOM-SRRH}.

	\subsection{Enhancement of the power minimization procedure through local power optimization (LPO)}\label{sec:NOMADBS-OPA}
	The power decrease incurred by a candidate subcarrier $n$ in the
	third phase of SRRH is greatly influenced by the
	amount of power $P_{k_2 , n ,r}$ allocated to user $k_2$ on $n$ using FTPA. Indeed, the addition of a new subcarrier translates into an increase of the power level allocated to the user on the one hand, and conversely into a power decrease for the same user due to the
	subsequent waterline reduction on its sole subcarriers on the other hand. Therefore, we propose to optimize the value
	of $P_{k_2 , n ,r}$ in such a way that the resulting user power reduction
	is minimized:
	\begin{equation*}
	\min_{P_{k_2,n,r}}\Delta P_{k_2}
	\end{equation*}
	Subject to:
	\begin{equation*}
	P_{k_2,n,r}\geq P_{k_1,n,r}
	\end{equation*}
	By replacing (\ref{equ:NOdichotomy}) into (\ref{DeltaPNOMA}), and expressing 
	$R_{k_2 ,n,r}$ using (\ref{rate2equation}), we can formulate the Lagrangian of this optimization problem as:
	\begin{align*}
	&L(P_{k_2,n,r},\lambda)=  P_{k_2,n,r} \\
	&+ N_{k_2}^{sole}w_{k_2}(N_{k_2}^{sole}) \left(
	\bigg(1+\frac{P_{k_2,n,r}h^2_{k_2,n,r}}{P_{k_1,n,r}h^2_{k_2,n,r}+\sigma^2}\bigg)^{\!\!\!-\frac{1}{N_{k_2}^{sole}}}\!-1\right)\\
	&+\lambda (P_{k_2,n,r}-P_{k_1,n,r})
	\end{align*}
	where $\lambda$ is the Lagrange multiplier.\\
	The corresponding Karush-Khun-Tucker (KKT) conditions are:
	\begin{equation*}
	\left\{\begin{split} 
	&\!\!1\!\!+\!\!\lambda
	\!\!-\!\!\frac{w_{k_2}(N_{k_2}^{sole})h^2_{k_2,n,r}}{P_{k_1,n,r}h^2_{k_2,n,r}+\sigma^2} \left(1+\frac{P^*_{k_2,n,r}h^2_{k_2,n,r}}{P_{k_1,n,r}h^2_{k_2,n,r}+\sigma^2}\right)^{\!\!\!\mkern-5mu \frac{-N_{k_2}-1}{N_{k_2}}}\!\!\!\!\!\!\!\!\!\!\!\!\!\!=\,\,\,0\\
	&\!\!\lambda (P^*_{k_2,n,r}-P_{k_1,n,r})=0 
	\end{split}
	\right.
	\end{equation*}
	We can check that the second derivative of the Lagrangian is always positive, and therefore the corresponding solution is the global minimum. For $\lambda=0$, this optimum is:
	\begin{multline}
	P^*_{k_2,n,r} = \left(\bigg(\frac{w_{k_2}(N_{k_2}^{sole})h^2_{k_2,n,r}}{P_{k_1,n,r}h^2_{k_2,n,r}+\sigma^2}\bigg)^{\frac{N^{sole}_{k_2}}{N^{sole}_{k_2}+1}}-1\right)\\\left(P_{k_1,n,r}+\frac{\sigma^2}{h^2_{k_2,n,r}}\right)
	\label{optimalpower}
	\end{multline} For $\lambda\neq 0$, $P^*_{k_2,n,r}=P_{k_1,n,r}$. However in such cases, with no power difference between the two paired users, successful SIC decoding is jeopardized at the receiver side for the first user. To overcome such a problem, we take:
	\begin{equation}
	P^*_{k_2,n,r} = P_{k_1,n,r}(1+\mu),\label{safety}
	\end{equation}
	with $\mu$ a positive safety power margin that depends on practical SIC implementation. In other terms, if the obtained $P^*_{k_2,n,r}$ in (\ref{optimalpower}) verifies the power constraint inequality, it is retained as the optimal solution, otherwise, it is taken as in (\ref{safety}). This method, referred to as ``SRRH-LP'', operates similarly to Algorithm \ref{alg:NOM-SRRH}, except for the FTPA power allocation which is replaced by either (\ref{optimalpower}) or (\ref{safety}). 
	
	\section{Resource Allocation Techniques for the Case of Two Powering RRHs per Subcarrier}
	Having extended previous CAS RA schemes to the DAS context and enhanced the corresponding solutions, the rest of the paper aims at designing specific NOMA RA schemes capturing the unique properties that arise in DAS. We start by developing the theoretical foundation lying behind SIC implementations when different RRHs are used to power the multiplexed signals on a subcarrier. The results show that under some well defined conditions, both paired users can perform SIC on the subcarrier. Finally, we propose several RA schemes taking advantage of the capacity gains inherent to mutual SIC and combine them with single SIC techniques.
	\subsection{Theoretical foundation}	\label{InformationTheory}
	In the case where the same RRH powers both multiplexed users on a subcarrier, there always exists one strong user at a given time which is the user having the best subcarrier-RRH link. However, this isn't necessarily the case when different RRHs are chosen to power the subcarrier, since the concept of weak and strong users is only valid relatively to a specific transmitting antenna. Indeed, the greater diversity provided by powering multiplexed subcarriers by different RRHs involves four instead of two different user-RRH links and thus opens the possibility of having more than one ``strong'' user at a time.
	
	\begin{theorem}
		Two users $k_1$ and $k_2$, paired on subcarrier $n$ and powered by two different RRHs, respectively $r_1$ and $r_2$, can both perform SIC if:
		\begin{align}
		h_{k_1,n,r_2} \geq 	h_{k_2,n,r_2}\label{equ:MSICcondition1}\\
		h_{k_2,n,r_1} \geq 	h_{k_1,n,r_1}\label{equ:MSICcondition2}
		\end{align}
	\end{theorem}
	
	\begin{proof}
		\textcolor{black}{Let $s_1$ be the signal of user $k_1$ emitted by RRH $r_1$ with a power $P_{k_1,n,r_1}$, and let $s_2$ be the signal of user $k_2$ emitted by RRH $r_2$ with a power $P_{k_2,n,r_2}$. Therefore, the channel conditions experienced by every signal arriving at a given user are different: at the level of $k_1$, the power levels of the signals $s_1$ and $s_2$ are $P_{k_1,n,r_1}h_{k_1,n,r_1}^2$ and $P_{k_2,n,r_2}h_{k_1,n,r_2}^2$ respectively. Similarly, at the level of $k_2$, the power levels of signals $s_1$ and $s_2$ are $P_{k_1,n,r_1}h_{k_2,n,r_1}^2$ and $P_{k_2,n,r_2}h_{k_2,n,r_2}^2$ respectively.} Depending on their respective signal quality, users $k_1$ and $k_2$ can decode the signal $s_2$ at different rates. Let $R_{k_2}^{(k_1)}$  be the necessary rate at the level of user $k_1$ to decode the signal of user $k_2$ in the presence of the signal of user $k_1$. And let $R_{k_2}^{(k_2)}$ the necessary rate to decode the signal of user $k_2$ at the level of $k_2$ in the presence of the signal of user $k_1$. The capacity that can be achieved by $k_1$ and $k_2$ over the signal $s_2$ and in the presence of the interfering signal $s_1$ are given by the Shannon limit:
		\begin{align}
		&R_{k_2}^{(k_1)} = \frac{B}{S} \log_2 \left(1+\frac{P_{k_2,n,r_2}h^2_{k_1,n,r_2}}{P_{k_1,n,r_1}h^2_{k_1,n,r_1}+\sigma^2}\right)\label{equ:channelConditionsU1}\\
		&R_{k_2}^{(k_2)} = \frac{B}{S}\log_2\left(1+\frac{P_{k_2,n,r_2}h^2_{k_2,n,r_2}}{P_{k_1,n,r_1}h^2_{k_2,n,r_1}+\sigma^2}\right)\label{equ:channelConditionsU2}
		\end{align}
		For $k_1$ to be able to perform SIC, the rates should satisfy the following condition:
		\begin{equation}
		R_{k_2}^{(k_1)} \geq R_{k_2}^{(k_2)} \label{equ:RcondUser1}
		\end{equation}
		By writing: $R_{k_2}^{(k_1)}-R_{k_2}^{(k_2)} = \frac{B}{S} \log_2\left(\frac{X}{Y}\right)$, we can express $X-Y$ as:
		\begin{align}
		X - Y = &P_{k_1,n,r_1}P_{k_2,n,r_2}\left(h^2_{k_1,n,r_2}h^2_{k_2,n,r_1}-h^2_{k_2,n,r_2}h^2_{k_1,n,r_1}\right)\nonumber\\
		+&\sigma^2 P_{k_2,n,r_2}\left(h^2_{k_1,n,r_2}-h^2_{k_2,n,r_2}\right) \label{equ:MSICuser1} 	
		\end{align}
		Similarly for user $k_2$, the rate condition that should be satisfied for the implementation of SIC at the level of $k_2$ is:
		\begin{equation}
		R_{k_1}^{(k_2)} \geq R_{k_1}^{(k_1)} \label{equ:RcondUser2} 
		\end{equation}
		$R_{k_1}^{(k_2)}$ and $R_{k_1}^{(k_1)}$ can be obtained from (\ref{equ:channelConditionsU1}) and (\ref{equ:channelConditionsU2}) by interchanging indexes 1 and 2. Also, by writing: $R_{k_1}^{(k_2)}-R_{k_1}^{(k_1)} = \frac{B}{S} \log_2\left(\frac{Z}{T}\right) $, we get:
		\begin{align}
		Z-T= &P_{k_2,n,r_2}P_{k_1,n,r_1}\left(h^2_{k_2,n,r_1}h^2_{k_1,n,r_2}-h^2_{k_1,n,r_1}h^2_{k_2,n,r_2}\right)\nonumber\\
		+&\sigma^2 P_{k_1,n,r_1}\left(h^2_{k_2,n,r_1}-h^2_{k_1,n,r_1}\right)   \label{equ:MSICuser2} 	
		\end{align} 
		Let us note that for the special case of $r_1=r_2=r$, we get:
		\begin{align*}
		X-Y&= \sigma^2 P_{k_2,n,r}\big(h^2_{k_1,n,r}-h^2_{k_2,n,r}\big)\\
		Z-T &=-\sigma^2 P_{k_1,n,r}\big(h^2_{k_1,n,r}-h^2_{k_2,n,r}\big)
		\end{align*}
		Therefore, either (\ref{equ:MSICuser1}) or (\ref{equ:MSICuser2}) is positive, not both, which justifies why only the stronger user, the one with the higher channel gain, is able to perform SIC as it has been stated in all previous works on NOMA \cite{R10,R9,R11,R15,R14}.
		
		For both users to perform SIC, the rate conditions (\ref{equ:RcondUser1}) and (\ref{equ:RcondUser2}) must be verified at the same time. By inspecting (\ref{equ:MSICuser1}) and (\ref{equ:MSICuser2}), we conclude that the following two conditions are sufficient to enable mutual SIC:
		\begin{gather*}
		h_{k_1,n,r_2} \geq 	h_{k_2,n,r_2}\\
		h_{k_2,n,r_1} \geq 	h_{k_1,n,r_1}
		\end{gather*}
		Indeed, these conditions ensure the positivity of each of the two terms in both $X-Y$ and $Z-T$. This concludes our proof.
	\end{proof}
	Regarding the power multiplexing constraints, the key is to design the power allocation scheme in such a way that the received power of the first signal to be decoded is larger than the power of the other signal. The resulting power conditions for users $k_1$ and $k_2$ respectively become:
	\begin{align*}
	P_{k_1,n,r_1}h^2_{k_1,n,r_1} \leq P_{k_2,n,r_2}h^2_{k_1,n,r_2}\\
	P_{k_2,n,r_2}h^2_{k_2,n,r_2} \leq P_{k_1,n,r_1}h^2_{k_2,n,r_1}
	\end{align*}
	They can be combined into the following condition:
	\begin{equation}
	\frac{h^2_{k_1,n,r_1}}{h^2_{k_1,n,r_2}} \leq \frac{P_{k_2,n,r_2}}{P_{k_1,n,r_1}} \leq \frac{h^2_{k_2,n,r_1}}{h^2_{k_2,n,r_2}}\label{equ:PowerConstraints}
	\end{equation}
	\begin{remark}\label{Remark}
		If (\ref{equ:MSICcondition1}) and (\ref{equ:MSICcondition2}) are true, then $\frac{h_{k_1,n,r_1}^2}{h_{k_1,n,r_2}^2} \leq \frac{h^2_{k_2,n,r_1}}{h^2_{k_2,n,r_2}}$. In this case, a  PA scheme can be found to allow a mutual SIC, i.e. there exist $P_{k_1,n,r_1}$ and $P_{k_2,n,r_2}$ such that (\ref{equ:PowerConstraints}) is true.
	\end{remark}
	
	Finally, the conditions (\ref{equ:MSICcondition1}) and (\ref{equ:MSICcondition2}) are sufficient but not necessary for the application of mutual SIC. Actually, the conditions for the application of mutual SIC lie in the positivity of (\ref{equ:MSICuser1}) and (\ref{equ:MSICuser2}). If any of (\ref{equ:MSICcondition1}) or (\ref{equ:MSICcondition2}) is not valid, the power terms in (\ref{equ:MSICuser1}) and (\ref{equ:MSICuser2}) should be considered, since they affect the sign of both equations. However, a closer examination of (\ref{equ:MSICuser1}) and (\ref{equ:MSICuser2}) reveals that in practical systems, their numerical values are greatly dominated by their first common term, since in general $\sigma^2 << Ph_{k,n,r}^2$. In that regard, a simpler constraint is derived on the channel gains:
	\begin{equation}
	h_{k_1,n,r_1}h_{k_2,n,r_2} \leq h_{k_2,n,r_1}h_{k_2,n,r_1}    \label{equ:FinalMutSIcCondition}
	\end{equation}
	This constraint will be used instead of (\ref{equ:MSICcondition1}) and (\ref{equ:MSICcondition2}) in the sequel. Note that condition (\ref{equ:FinalMutSIcCondition}) also ensures the existence of a PA scheme that will allow a mutual SIC. When both users $k_1$ and $k_2$ perform SIC on a subcarrier $n$, their reachable rates on $n$ are given by:
	\begin{align}
	R_{k_1,n,r_1} &= \frac{B}{S}\log_2\left(1+\frac{P_{k_1,n,r_1}h_{k_1,n,r_1}^2}{\sigma^2}\right)\label{rateSIC1}\\
	R_{k_2,n,r_2} &= \frac{B}{S}\log_2\left(1+\frac{P_{k_2n,r_2}h_{k_2,n,r_2}^2}{\sigma^2}\right)\label{rateSIC2}
	\end{align}
	Following the introduction of mutual SIC, the RA strategy should be modified accordingly. Therefore, the next sections describe the
	development of novel RA techniques that can \mbox{benefit} from this new potential of the NOMA-DAS combination.
	\subsection{Mutual SIC-based power minimization without power multiplexing constraints}
	The new RA problem in hand is still combinatorial, which motivates the proposal of suboptimal RA schemes in the following sections. 
	
	In addition to the selection of different antennas in the pairing phase of Algorithm \ref{alg:NOM-SRRH}, the key modifications that must be accounted for, when moving from single SIC to mutual SIC RA schemes, involve:
	\begin{itemize}
		\item Subcarrier subset selection: only the subcarrier-RRH links satisfying (\ref{equ:FinalMutSIcCondition}) are considered for potential assignment in mutual SIC configurations.
		\item Power assignment: the power multiplexing constraint (\ref{equ:PowerConstraints}) must be accounted for.
	\end{itemize} 
	
	We first address a relaxed version of the problem where the power multiplexing constraints are disregarded. This consideration reverts the optimal PA scheme in the pairing phase to the user-specific waterfilling solution in OMA. Therefore, the pairing phase in mutual SIC becomes a simple extension of the OMA phase from Algorithm \ref{alg:NOM-SRRH}. This method referred to as MutSIC-UC will be used as a lower bound on the performance of mutual SIC algorithms (in terms of the total transmit power).
	
	To compensate for the disregarded constraints, subcarrier assignment should be followed by a power optimization step as shown in Appendix \ref{app:2}. However, the set of possible power corrections grows exponentially with the number of multiplexed subcarriers. Therefore, alternative suboptimal strategies accounting for the power multiplexing constraints at every subcarrier assignment are investigated in the following sections.
	
	\subsection{Mutual SIC power minimization with direct power adjustment (DPA)}
	From a power minimization perspective, the power distribution obtained through waterfilling is the best possible PA scheme. However, compliance with the power multiplexing conditions is not guaranteed; therefore, a power adjustment should occasionally be made on the multiplexed subcarriers.
	
	When an adjustment is needed, the new value of $P_{k_2,n,r_2}$ in (\ref{equ:PowerConstraints}) should fall between $P_{k_1,n,r_1}h^2_{k_1,n,r_1}/h^2_{k_1,n,r_2}$ and $P_{k_1,n,r_1}h^2_{k_2,n,r_1}/h^2_{k_2,n,r_2}$ (the value of $P_{k_1,n,r_1}$ is fixed). However, since any deviation from the waterfilling procedure degrades the performance of the solution, this deviation must be rendered minimal. Therefore, $P_{k_2,n,r_2}$ is set at the nearest limit of the inequality (\ref{equ:PowerConstraints}), with some safety margin $\mu$ accounting for proper SIC decoding. After this adjustment, the powers on the multiplexed subcarrier are kept unvaried, as in Algorithm \ref{alg:NOM-SRRH}. This procedure will be referred to as MutSIC-DPA; its details are presented in Algorithm \ref{alg:NOMA-DPA}.
	
	\begin{algorithm}[h]
		\caption{MutSIC-DPA}
		\label{alg:NOMA-DPA}
		\textbf{\small Phase 1:}\\
		\small Worst-Best-H followed by OMA single-user assignment\\\\
		\textbf{Phase 2:} \textcolor{gray}{\small // NOMA MutSIC pairing}\\
		\[k_2 = \mathop{\arg\max}_k P_{k,tot}    \]
		\[S_c = \{(n,r_2)\text{ s.t. (\ref{equ:FinalMutSIcCondition}) \& (\ref{Omachannelcondition}) are verified}\}\]
		\textbf{\small For} every candidate couple $(n,r_2) \in S_c$\\
		\text{\quad Calculate $P_{k_2,n,r_2}^*$ and $\Delta P_{k_2,n,r_2}$} using (\ref{waterdrop}) and (\ref{DeltaPOMA})\\
		\[\text{\quad \textbf{If} } P_{k_2,n,r_2}^* \text{verifies (\ref{equ:PowerConstraints}), set }P_{k_2,n,r_2}=P_{k_2,n,r_2}^* \] 
		\[\text{\quad \textbf{If} }\frac{P_{k_2,n,r_2}}{P_{k_1,n,r_1}} < \frac{h_{k_1,n,r_1}^2}{h_{k_1,n,r_2}^2} \text{ set }P_{k_2,n,r_2} \!=\! (1+\mu)P_{k_1,n,r_1}\!\frac{h_{k_1,n,r_1}^2}{h_{k_1,n,r_2}^2} \]
		\[\text{\quad and estimate } \Delta P_{k_2,n,r_2} \text{ using (\ref{equ:NOdichotomy}) and (\ref{DeltaPNOMA})}\]
		\[\text{\quad \textbf{If} } \frac{P_{k_2,n,r_2}}{P_{k_1,n,r_1}} > \frac{h_{k_2,n,r_1}^2}{h_{k_2,n,r_2}^2}\text{ set } P_{k_2,n,r_2} \!=\! (1-\mu)P_{k_1,n,r_1}\!\frac{h_{k_2,n,r_1}^2}{h_{k_2,n,r_2}^2}\]
		\[\text{\quad and estimate } \Delta P_{k_2,n,r_2} \text{ using (\ref{equ:NOdichotomy}) and (\ref{DeltaPNOMA})} \]
		\textbf{End for}\\
		\[(n^*,r^*_2) = \mathop{\arg\min}_{(n,r_2)} \Delta P_{k_2,n,r_2}\]
		Continue similarly to SRRH
	\end{algorithm}
	\subsection{Mutual SIC power minimization with sequential optimization for power adjustment (SOPA)}
	In order to improve on the MutSIC-DPA
	technique, we propose to replace the adjustment and power
	estimation steps by a sequential power optimization. Instead of
	optimizing the choice of $P_{k _2 , n , r_2}$ over the candidate couple $(n, r_2 )$, we look for a wider optimization in which powers of both first
	and second users on the considered subcarrier are adjusted, in a
	way that their global power variation is minimal:
	\begin{align*}
	\{P_{k_1,n,r_1},P_{k_2,n,r_2}\}^* = \!\!\!\!\mathop{\arg\max}_{P_{k_1,n,r_1},P_{k_2,n,r_2}}\!\!\!\!(-\Delta P_{k_1,n,r_1}-\Delta P_{k_2,n,r_2})
	\end{align*}
	subject to:
	\begin{align*} 
	\frac{h_{k_1,n,r_1}^2}{h_{k_1,n,r_2}^2}\leq \frac{P_{k_2,n,r_2}}{P_{k_1,n,r_1}}\\
	\frac{P_{k_2,n,r_2}}{P_{k_1,n,r_1}} \leq \frac{h_{k_2,n,r_1}^2}{h_{k_2,n,r_2}^2}
	\end{align*}
	The power variations of users $k_2$ and $k_1$ are given by:
	\begin{align*}
	\Delta P_{k_2,n,r_2}\!\!\!= &N^{sole}_{k_2}w_{k_2}\!(N^{sole}) \!\!\left(\!\!\!\left(1+\frac{P_{k_2,n,r_2}h_{k_2,n,r_2}^2}{\sigma^2}\right)^{\!\!\!-\frac{1}{N_{k_2}^{sole}}}\!\!\!\!\!\!\!\!\!-1\!\right)\nonumber\\
	&\!\!\!+P_{k_2,n,r_2} 
	\end{align*}
	\begin{align*}
	\Delta P_{k_1,n,r_1} &= (N_{k_1}^{sole}-1)W_{I,k_1}\left( 2^{-\frac{\Delta R_{k_1}S}{(N_{k_1}^{sole}-1)B}}-1 \right)\!+\! P_{k_1,n,r_1}\\&-P_{k_1,n,r_1}^I
	\end{align*}
	where $P^I_{k_1,n,r_1}$ is the initial power allocated on $n$ to $k_1$ and $W_{I,k_1}$ the initial waterline of $k_1$ (before pairing with user $k_2$). Also, the rate variation of user $k_1$ on $n$, due to the power adjustment on $n$, can be written as:
	\[\Delta R_{k_1} = \frac{B}{S}\log_2\left(\frac{\sigma^2+P_{k_1,n,r_1}h^2_{k_1,n,r_1}}{\sigma^2+P^I_{k_1,n,r_1}h^2_{k_1,n,r_1}}\right) \]
	The Lagrangian of this problem is:
	\begin{multline*}
	L(P_{k_1,n,r_1},P_{k_2,n,r_2},\lambda_1,\lambda_2) \!\!=\!\!-\lambda_1\!\!\left(\!\! P_{k_1,n,r_1}\frac{h^2_{k_1,n,r_1}}{h^2_{k_1,n,r_2}}-P_{k_2,n,r_2} \!\!\right)\\- \lambda_2\left( P_{k_2,n,r_2}-P_{k_1,n,r_1}\frac{h^2_{k_2,n,r_1}}{h^2_{k_2,n,r_2}}\right)-\Delta P_{k_1,n,r_1}-\Delta P_{k_2,n,r_2}  
	\end{multline*}
	The solution of this problem must verify the following conditions:
	\begin{align*}
	&\nabla L (P_{k_1,n,r_1},P_{k_2,n,r_2},\lambda_1,\lambda_2)=0\\
	&\lambda_1\left( P_{k_1,n,r_1}h^2_{k_1,n,r_1}/h^2_{k_1,n,r_2}-P_{k_2,n,r_2} \right) = 0\\
	& \lambda_2\left( P_{k_2,n,r_2}-P_{k_1,n,r_1}h^2_{k_2,n,r_1}/h^2_{k_2,n,r_2}\right) =0\\
	& \lambda_1, \lambda_2 \geq 0 
	\end{align*}
	Four cases are identified:
	\begin{align*}
	&\text{1. }\lambda_1 = 0,\lambda_2=0\\
	&\text{2. }\lambda_1 \ne 0,\lambda_2=0 \rightarrow P_{k_2,n,r_2}=P_{k_1,n,r_1}h^2_{k_1,n,r_1}/h^2_{k_1,n,r_2}\\
	&\text{3. } \lambda_1=0, \lambda_2 \ne 0 \rightarrow P_{k_2,n,r_2}=P_{k_1,n,r_1}h^2_{k_2,n,r_1}/h^2_{k_2,n,r_2}\\
	&\text{4. } \lambda_1 \ne 0, \lambda_2 \ne 0
	\end{align*}
	Case 1 corresponds to the unconstrained waterfilling solution applied separately to the two users. Case 4 is generally
	impossible, since the two boundaries of the inequality (\ref{equ:PowerConstraints}) would be equal. Considering case 2, by replacing $P_{k_2,n,r_2}$ in terms of $P_{k_1,n,r_1}$ in the Lagrangian and by taking the derivative with respect to $P_{k_1,n,r_1}$, we can verify that $P_{k_1,n,r_1}^*$ is the solution of the following nonlinear equation:
	\begin{multline}\label{equ:nonlinearEqation}
	W_{I,k_2}\frac{h^2_{k_1,n,r_1}h^2_{k_2,n,r_2}}{h^2_{k_1,n,r_2}\sigma^2} \!\!\left( 1+ \frac{P_{k_1,n,r_1}h^2_{k_2,n,r_2}h^2_{k_1,n,r_1}}{\sigma^2} \right)^{\!\!\!\!-\frac{1}{N^{sole}_{k_2}}-1}\\+
	\frac{W_{I,k_1}h^2_{k_1,n,r_1}}{\sigma^2+P^I_{k_1,n,r_1}h^2_{k_1,n,r_1}}\left( \frac{\sigma^2+P_{k_1,n,r_1}h^2_{k_1,n,r_1}}{\sigma^2+P^I_{k_1,n,r_1}h^2_{k_1,n,r_1}} \right)^{-\frac{1}{N_{k_1}^{sole}-1}-1}\\ - \frac{h^2_{k_1,n,r_1}}{h^2_{k_1,n,r_2}}-1  	=0
	\end{multline}
	Note that in practice, we also take into consideration the safety
	power margin $\mu$ in the calculation of $P_{k_1,n,r_1}$. Similar calculations
	are performed for case 3. The solution that yields the lowest $\Delta P$ is retained. Also, if none of the cases provides positive power
	solutions, the current candidate couple $(n, r_2 )$ is discarded. This
	method of optimal power adjustment (OPAd) is employed both at the subcarrier allocation stage (for the selection of the best candidate couple $\
	(n,r_2)$ for user $k_2$) and at the power allocation stage (following the selection of the subcarrier-RRH pair). It will be referred to
	as ``MutSIC-OPAd''.\\
	
	Finally, in order to decrease the complexity of ``MutSIC-OPAd'', inherent to the resolution of a nonlinear
	equation for every subcarrier-RRH candidate, we consider a ``semi-optimal'' variant of this technique,
	called ``MutSIC-SOPAd'': at the stage where
	candidate couples $(n, r_2 )$ are considered for potential assignment
	to user $k_2$, DPA is used for power adjustment to
	determine the best candidate in a cost-effective way. Then, the
	preceding OPAd solution is applied to allocate power levels to
	users $k_1$ and $k_2$ on the retained candidate.
	
	\subsection{Combination of the allocation of mutual and single SIC subcarriers in DAS}
	
	To further exploit the space diversity inherent to DAS and minimize the system transmit power, single SIC and mutual SIC algorithms are combined to take advantage of the full potential of NOMA techniques. Given the superiority of mutual SIC over single SIC schemes, we prioritize the allocation of subcarriers allowing mutual SIC by first applying MutSIC-SOPAd. Then, the remaining set of
	solely assigned subcarriers is further examined for potential
	allocation of a second user in the single SIC context, using the same
	RRH as that of the first assigned user. LPO is used for power allocation in this
	second phase. This method will be referred to as ``Mut\&SingSIC''.

	\section{Complexity Analysis}\label{Sec:Complexity}
	In this section, we analyze the complexity of the different
	allocation techniques proposed in this study. The complexity of
	OMA-CAS, NOMA-CAS and OMA-DAS is also considered for comparison. It is studied by
	considering an implementation that includes the runtime
	enhancement procedures introduced in section \ref{sec:RuntimeEnhancment}. In OMA-CAS and OMA-DAS scenarios, only phases 1 and 2 of Algorithm \ref{alg:NOM-SRRH} are applied, with either
	$R$=1 (for OMA-CAS) or $R \neq1$ (for OMA-DAS). In NOMA-CAS, Algorithm \ref{alg:NOM-SRRH} is used with $R$=1.
	
	For OMA-CAS, we consider that the channel matrix
	is reordered so that the subcarriers for each user are sorted by
	the decreasing order of channel gain. This step accelerates
	the subsequent subcarrier allocation stages and has a complexity of
	$O(KS\log(S))$. Following the Worst-Best-H phase, each iteration
	complexity is mainly dominated by the search of the most power
	consuming user with a cost $O(K)$. Assuming all the remaining $S-K$ subcarriers are allocated, the final complexity is
	$O(KS\log(S)+(S-K)K)$.
	
	Each allocation step in the pairing phase of NOMA-CAS
	consists of the identification of the most power consuming user,
	followed by a search over the subcarrier space, and a power
	update over the set of sole subcarriers for the user, with an average
	number of $S/K$ subcarriers. Assuming $S$ paired subcarriers, the total complexity of
	NOMA-CAS is $O(KS\log(S)+(S-K)K+S(K+S+S/K))$.
	
	In OMA-DAS, we consider an initial sorting of each user
	subcarrier gains, separately for each RRH, with a cost of
	$O(KSR\log(S))$. Then, an allocation cycle consists of user
	identification, followed by the search of the RRH providing the
	subcarrier with the highest channel gain. This corresponds to a
	complexity of $O(K+R)$. Therefore, the total complexity is:
	$O(KSR\log(S)+(S-K)(K+R))$. Consequently, the total complexity
	of SRRH and SRRH-LPO is
	$O(KSR\log(S)+(S-K)(K+R)+S(K+S+S/K))$.
	\textcolor{black}{In order to assess the efficiency of SRRH-LPO, we compare our solution to the optimal power allocation technique developed in \cite{R14}. More specifically, we apply SRRH-LPO to determine the user-subcarrier-RRH assignment; then, in a second phase, we apply the optimal PA in \cite{R14}. This technique will be referred to as SRRH-OPA; its complexity analysis and comparison with SRRH-LPO is provided in appendix \ref{AppendixComplexity}.}

	Concerning MutSIC-UC, by following the same
	reasoning as for OMA-DAS, and accounting for the search of an
	eventual collocated user for at most $S$ subcarriers, we get a total
	of $O(KSR\log(S)+(S-K)(K+R)+S(K+R-1))$.
	
	As for MutSIC-DPA, the total complexity is
	$O(KSR\log(S)+(S-K)(K+R)+S(K+S(R-1)+S/K))$, where the $S(R-
	1)$ term stems from the fact that the search over the subcarrier
	space in the pairing phase is conducted over all combinations of
	subcarriers and RRHs, except for the RRH of the first user on the
	candidate subcarrier.
	
	Regarding MutSIC-OPAd, let $C$ be the complexity of
	solving the nonlinear equation (\ref{equ:nonlinearEqation}). The total complexity is
	therefore
	$O(KSR\log(S)+(S-K)(K+R)+S(K+S(R-1)C+S/K))$.
	Given that MutSIC-SOPAd solves (\ref{equ:nonlinearEqation}) only once
	per allocation step, its complexity is $O(KSR\log(S)+(S-
	K)(K+R)+S(K+S(R-1)+S/K+C))$. Consequently, the complexity
	of Mut\&SingSIC is $O(KSR\log(S)+(S-
	K)(K+R))+S(K+S(R-1)+S/K+C)+S(K+S+S/K)$. The additional
	term corresponds to the Single SIC phase which is similar to the
	pairing phase in NOMA-CAS.
	\begin{table}[H]
		\caption{Approximate complexity of the different allocation techniques.}
		\label{tab:1}
		\begin{center}
			\setlength\tabcolsep{1pt}
			\begin{tabular}{|l | c| c|}
				\hline
				\rowcolor{lightgray} \multicolumn{1}{|c|}{\textbf{Technique}}  & \textbf{Complexity} \\
				\hline
				OMA-CAS & $O( KS \log( S ))$ \\
				\hline
				NOMA-CAS & $O( S^2 + KS\log( S )) $ \\
				\hline
				OMA-DAS & $O(KSR\log(S)) $\\
				\hline
				SRRH & $O( S^2 + KSR \log( S ))$\\
				\hline
				SRRH-LPO  & $	O( S^2 + KSR \log( S ))$ \\
				\hline
				MutSIC-UC & 	$O( KSR \log( S ))$ \\
				\hline
				MutSIC-DPA & $O( S^2 R+KSR \log( S ))$ \\
				\hline
				MutSIC-OPAd & $O( S^2 RC+KSR \log( S ) )$ \\
				\hline
				MutSIC-SOPAd &$O( S^2 R + SC+KSR \log( S ) )$ \\
				\hline
				Mut\&SingSIC & $O( S^2 R + SC+KSR \log( S ) )$ \\
				\hline
			\end{tabular}
		\end{center}
	\end{table}	
	To give an idea of the relative complexity orders, we present in Table \ref{tab:1} the approximate complexity of the different techniques. In fact, the complexity of the methods employing a numerical solver depends on the resolution cost $C$ that is dependent on the closeness of the initial guess to the actual solution. In that regard, we note that MutSIC-SOPAd is roughly $C$ times less complex than MutSIC-OPAd, and has a complexity comparable to MutSIC-DPA.
	
	\section{Performance Results}
	The performance of the different allocation techniques are
	assessed through simulations in the LTE/LTE-
	Advanced context \cite{ETSI_ChannelCoding_LTE_2011}. The cell is hexagonal with
	an outer radius $R_d$ of 500 m. For DAS, we consider four RRHs ($R=4$), unless specified otherwise. One antenna is located at the
	cell center, while the others are uniformly positioned
	on a circle of radius 2$R_d/$3 centered at the cell center. The number of users in the cell is $K=15$, except for Fig. \ref{fig:5}. The system
	bandwidth $B$ is 10 MHz, and it is divided into $S=64$ subcarriers except for Fig. \ref{fig:5}. The transmission medium is a
	frequency-selective Rayleigh fading channel with a root mean square delay spread of 500
	ns. We consider distance-dependent path-loss with a decay factor
	of 3.76 and lognormal shadowing with an 8 dB variance. The
	noise power spectral density $N_0$ is $4.10^{-18}$ mW/Hz. In this study,
	we assume perfect knowledge of the user channel gains by the
	BBU. \textcolor{black}{For typical system parameters, the system performance in terms of transmit power is mainly invariant with $\rho$, thus $\rho$ is set to $10^{-3}~W$. A detailed analysis of the system behavior in terms of $\rho$ can be found in \cite{oma} for OMA systems. The $\alpha$ decay factor in FTPA is taken equal to 0.5,and the safety power
		margin $\mu$ is set to 0.01. The performance results of OMA-CAS, NOMA-CAS and OMA-DAS are also shown for comparison.}\\

	Fig. \ref{fig:2} represents the total transmit power in the cell in terms
	of the requested rate considering only SRRH schemes for NOMA-based techniques.
	\begin{figure}[H]
	\includegraphics[width=\textwidth/2]{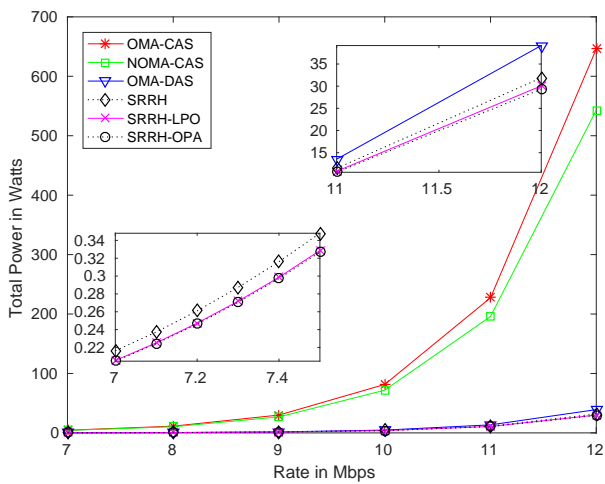}
		\caption{Total power in terms of $R_{k,req}$ for DAS and CAS scenarios, with OMA and NOMA-SRRH schemes}
		\label{fig:2}
	\end{figure}
	The results show that the DAS configuration greatly
	outperforms CAS: a large leap in power with a factor around 16
	is achieved with both OMA and NOMA signaling. At a target
	rate of 12 Mbps, the required total power using SRRH, SRRH-LPO and SRRH-OPA is respectively 17.6\%, 24.5\%, and 26.1\% less than in OMA-DAS. This shows a clear advantage of NOMA over OMA in the
	DAS context. Besides, applying LPO allows a power reduction
	of 7.7\% over FTPA, with a similar computational load. \textcolor{black}{The
		penalty in performance of LPO with respect to optimal PA is only 2\% at 12 Mbps, but with a greatly reduced complexity.}
	
	In Fig. \ref{fig:3}, the results are focused on the evaluation of mutual SIC and
	single SIC configurations. It can be seen that all three
	constrained configurations based on pure mutual SIC (MutSIC-DPA,
	MutSIC-SOPAd
	and
	MutSIC-OPAd) largely outperform SRRH-LPO. Their gain towards the latter is respectively 56.1\%, 63.9\% and 72.9\%, at a requested rate of 13 Mbps. The
	significant gain of optimal power adjustment towards its
	suboptimal counterpart comes at the cost of a significant
	complexity increase, as shown in Section \ref{Sec:Complexity}. The most power-efficient mutual SIC implementation is obviously MutSIC-UC, since it is designed to solve a relaxed version of the
	power minimization problem by dropping all power
	multiplexing constraints. Therefore, it only serves as a benchmark for assessing the other methods, because power
	multiplexing conditions are essential for allowing correct signal
	decoding at the receiver side.
	\begin{figure}[h]
		\includegraphics[width=\textwidth/2]{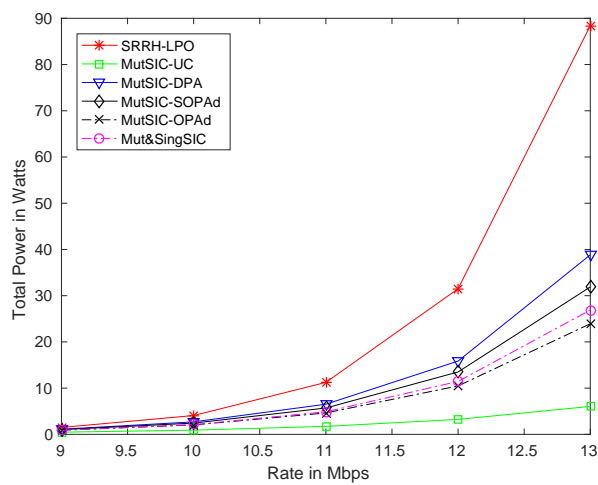}
		\caption{Total power in terms of $R_{k,req}$ for the proposed NOMA-DAS schemes}
		\label{fig:3}
	\end{figure}
	
	Except for the OPAd solution, the best global strategy remains the combination of mutual
	and single SIC subcarriers, since it allows a power reduction of 15.2\% and 15.6\% at 12 and 13 Mbps respectively, when compared to  
	MutSIC-SOPAd.
	
	Fig. \ref{fig:4} shows the influence of increasing the number of
	RRHs on system performance. As expected, increasing the
	number of spread antennas greatly reduces the overall power, either with single SIC or combined mutual and single SIC configurations. A significant power reduction is observed when $R$ is increased from 4 to 5, followed by a more moderate
	one when going from 5 to 7 antennas. The same behavior is
	expected for larger values of $R$. However, practical considerations like the overhead of CSI signaling exchange and the synchronization of the distributed RRHs, not to mention geographical deployment constraints, would suggest limiting the number of deployed
	antennas in the cell.
	\begin{figure}[h]
		\includegraphics[width=\textwidth/2]{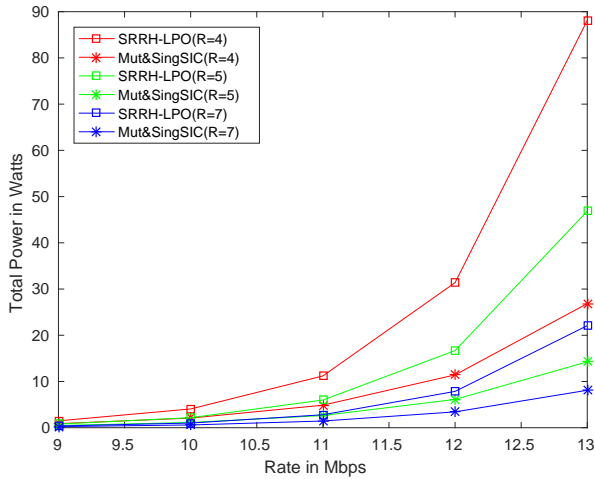}
		\caption{Total power in terms of $R_{k,req}$ for NOMA-DAS schemes, with $K$=15, $S$=64, and $R$=4, 5 or 7}
		\label{fig:4}
	\end{figure}
	
	In Fig. \ref{fig:5}, we show the performance for a varying number of
	users, for the case of 4 RRHs and 128 subcarriers. Results confirm that the allocation strategies based on mutual
	SIC, or combined mutual and single SIC, scale much better to
	crowded areas, compared to single SIC solutions. The power
	reduction of Mut\&SingSIC towards SRRH-LPO is 69.8\% and 78.2\% for 36 and 40 users
	respectively.
	\begin{figure}[h]
		\includegraphics[width=\textwidth/2]{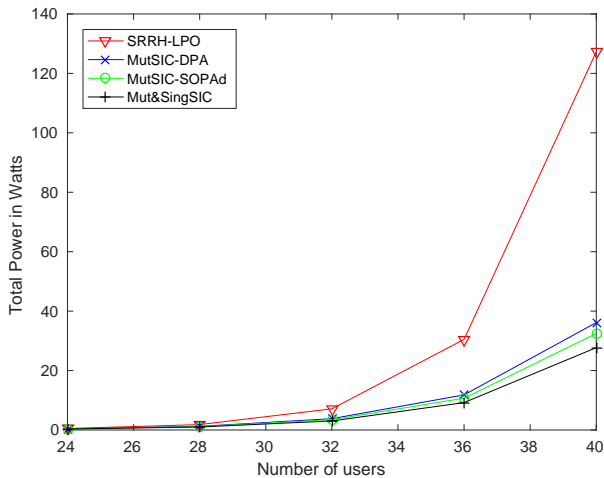}
		\caption{Total power in terms of the number of users for the NOMA-DAS schemes, with $R_{k,req}$=5 Mbps, $S$=128, and $R$=4}
		\label{fig:5}
	\end{figure}
	
	Table \ref{table:2} shows the statistics of the number of non-
	multiplexed subcarriers, the number of subcarriers where a
	mutual SIC is performed, and the number of subcarriers where a
	single SIC is performed. On average, SRRH-LPO uses single SIC
	NOMA on 25\% (resp. 32\%) of the subcarriers for $R _{k,req}$ = 9 Mbps
	(resp. 12 Mbps), while the rest of the subcarriers is solely allocated to users (a small proportion is not allocated at	all, depending on the power threshold $\rho$). On the other hand, the
	proportions are respectively 17\% and 23\% with MutSIC-SOPAd. Therefore, in light of the results of Figs.
	\ref{fig:3} and \ref{fig:5}, MutSIC-SOPAd not only outperforms SRRH-LPO from the requested transmit power perspective, but it
	also presents the advantage of yielding a reduced complexity at
	the User Equipment (UE) level, by requiring a smaller amount
	of SIC procedures at the receiver side. This shows the efficiency
	of the mutual SIC strategy, combined with appropriate power
	adjustment, over classical single SIC configurations.\\
	
	\begin{table}[h]
		\caption{Statistics of subcarrier multiplexing, for $K$=15, $S$=64, and $R$=4.}
		\begin{center}
			\label{table:2}
			\setlength\tabcolsep{6pt}
			\begin{tabular}{|l|c|c|c|}
				\hline
				\rowcolor{lightgray} \multicolumn{1}{|c|}{\textbf{Resource allocation}}  & \textbf{Non}  & \textbf{SC}  & \textbf{SC} \\
				\rowcolor{lightgray} \multicolumn{1}{|c|}{\textbf{technique}} & \textbf{Mux}& \textbf{MutSIC} & \textbf{SingSC}\\
				\rowcolor{lightgray} & \textbf{SC} & & \\
				\hline
				& \multicolumn{3}{c|}{$R_{k,req}$=9Mbps}\\
				\hline
				SRRH-LPO & 48.1 & - & 15.9\\
				\hline
				MutSIC-SOPAd & 53.4 & 10.6 & -\\
				\hline
				Mut\&SingSIC & 39.2 & 10.6 & 14.2\\
				\hline
				& \multicolumn{3}{c|}{$R_{k,req}$=12Mbps}\\
				\hline
				SRRH-LPO & 43.7 & - & 20.3\\
				\hline
				MutSIC-SOPAd & 49.4 & 14.6 & -\\
				\hline
				Mut\&SingSIC & 29 & 14.6 & 20.4\\
				\hline
				
			\end{tabular}
		\end{center}
	\end{table}
	It can be noted that in Mut\&SingSIC, 17\% (resp. 23\%) of the subcarriers are powered from different antennas. This shows the importance of
	exploiting the additional spatial diversity, combined with
	NOMA, inherent to DAS.
	
	\section{Conclusion}
	In this paper, various RA techniques were
	presented for minimizing the total downlink transmit power in
	DAS for 5G and beyond networks. We first proposed
	several enhancements to a previously developed method for CAS, prior to extending it to the DAS context.
	Furthermore, we unveiled some of the hidden potentials of DAS for
	NOMA systems and developed new techniques to make the most
	out of these advantages, while extracting their best
	characteristics and tradeoffs. Particularly, this study has enabled
	the design of NOMA with SIC decoding at both paired UE sides. Simulation results have shown the superiority of the proposed
	methods with respect to single SIC configurations. They also
	promoted mutual SIC with suboptimal power adjustment to the best tradeoff between transmit power and complexity at both the
	BBU and the UE levels. Several aspects of this work can be
	further explored, since many additional challenges need to be
	addressed to enhance the NOMA-DAS-specific resource
	allocation schemes. For instance, practical considerations can be incorporated in the
	study, such as imperfect antenna synchronization and limited
	CSI exchange. Furthermore, the study can be enriched by
	the use of MIMO antenna systems in a distributed context.  
	
	\section{Acknowledgment}
	This work has been funded with support from the Lebanese
	University and the PHC CEDRE program (research project between the Lebanese University and IMT Atlantique). Part of this
	work has been performed in the framework of the Horizon 2020
	project FANTASTIC-5G (ICT-671660), which is partly funded
	by the European Union.
	
	\begin{appendices}
		\section{}\label{app:1}
		\begin{proof}
			By inspecting (\ref{waterdrop}) and (\ref{DeltaPOMA}), we see that a higher channel gain ensures a lower waterline, but also a lower inverse channel gain in the expression of $\Delta P_{k,n_a,r}$. In order to prove that the subcarrier which yields the lowest $\Delta P$ is indeed the one with the best channel gain, we start by combining (\ref{waterdrop}) and (\ref{DeltaPOMA}), and we express the power decrease as:
			\begin{equation*}
			\Delta P_{k,n_a,r} = (N_k+1)\left(\frac{(w_k(N_k))^{N_k}}{h^2_{k,n_a,r}/\sigma^2}\right)^{\!\!\!1/(N_k+1)}\w\w\w\w-N_kw_k(N_k)-\frac{\sigma^2}{h^2_{k,n_a,r}}
			\end{equation*}
			By taking the derivative of $\Delta P_{k,n,r}$ with respect to $h_{k,n_a,r}$, we get:
			\begin{equation*}
			\frac{\partial\Delta P_{k,n,r}}{\partial h_{k,n_a,r}}=-2\frac{(\sigma^2)^{1/N_k+1}}{(h_{k,n_a,r})^{\frac{2}{N_k+1}+1}}\left(w_k(N_k)\right)^{N_k/(N_k+1)}+\frac{2\sigma^2}{h^3_{k,n_a,r}}
			\end{equation*}
			Therefore, we can verify that:
			\begin{equation*}
			\frac{\partial\Delta P_{k,n,r}}{\partial h_{k,n_a,r}} \leq 0 \Longleftrightarrow \frac{\sigma^2}{h^2_{k,n_a,r}} \leq \left(\frac{\sigma^2}{h^2_{k,n_a,r}}\right)^{1/N_k+1}\!\!\!\!\!\!\!\!\!\!\!\!\!\!\! \big(w_k(N_k)\big)^{N_k/(N_k+1)}
			\end{equation*}
			which directly leads to (\ref{Omachannelcondition}).\\
			We can deduce that, provided that (\ref{Omachannelcondition}) is verified by subcarrier $n_a$, $\Delta P_{k,n_a,r}$ is a monotonically decreasing function of $h_{k,n_a,r}$, which concludes the proof.
		\end{proof}
		
\section{Formulation of the Power Optimization Problem for the Constrained Case in Mutual SIC}
		\label{app:2}
		For a predefined subcarrier-RRH-user assignment, the constrained power minimization problem for power assignment can be cast as the solution of the following optimization problem:
		\mathleft
		\begin{align*}
		&\max_{\{P_{k,n,r}\}}\left(-\sum^K_{k=1} \sum^S_{n=1} \sum^R_{r=1} P_{k,n,r}\right),\\
		&\text{ Subject to:} \\
		&\sum_{n \in \mathcal{S}_k} \log_2 \left(1+\frac{P_{k,n,r}h^2_{k,n,r}}{\sigma^2}\right)=R_{k,req}, 1 \leq k \leq K\\
		&-\frac{P_{k_2,n,r_2}}{P_{k_1,n,r_1}} \leq -\frac{h^2_{k_1,n,r_1}}{h^2_{k_1,n,r_2}}, \forall n \in \mathcal{S}_{mSIC}\\
		&\frac{P_{k_2,n,r_2}}{P_{k_1,n,r_1}} \leq \frac{h^2_{k_2,n,r_1}}{h^2_{k_2,n,r_2}}, \forall n \in \mathcal{S}_{mSIC}\\
		\end{align*}
		where $\mathcal{S}_{mSIC}$ is the set of subcarriers undergoing a mutual SIC. The corresponding Lagrangian with multipliers $\lambda_k$ and $\beta_{i,n}$ is:
		\begin{align*}
		&	L(P,\lambda,\beta_1,\beta_2)=-\sum^K_{k=1} \sum^S_{n=1} \sum^R_{r=1} P_{k,n,r}\\
		&+\sum_{n \in S_{mSIC}}\beta_{1,n}\left(\frac{h^2_{k_2,n,r_1}}{h^2_{k_2,n,r_2}}-\frac{P_{k_2,n,r_2}}{P_{k_1,n,r_1}}\right)\\
		& + \sum_{n \in \mathcal{S}_{mSIC}}\beta_{2,n}\left(\frac{h^2_{k_2,n,r_2}}{h^2_{k_1,n,r_1}}-\frac{P_{k_1,n,r_1}}{P_{k_1,n,r_2}}\right)\\
		&+\sum_{k=1}^{K}\lambda_k \left(R_{k,req}- \sum_{n \in \mathcal{S}_k} \log_2\left(1+ \frac{P_{k,n,r}h^2_{k,n,r}}{\sigma^2}\right)\right)
		\end{align*}
		Writing the KKT conditions (not presented here for the sake of
		concision) leads to a system of $N_e$ non-linear equations with $N_e$ variables, where $N_e = 3 \vert S_{mSIC}\vert+K+S$ (taking into account the $S\!-\!\vert \mathcal{S}_{mSIC}\vert$ power variables on non-paired subcarriers). Knowing that $\beta_{1,n}$ and $\beta_{2,n}$ cannot be simultaneously non-zero, we have, for every subcarrier allocation scheme, a total of $3^{\vert \mathcal{S}_{mSIC}\vert}$ different possible combinations to solve, that is $3^{\vert \mathcal{S}_{mSIC}\vert}$ different variations of a square system of $2\vert \mathcal{S}_{mSIC}\vert+ K + S$ equations (per subcarrier allocation).
		\section{Complexity Analysis of SRRH-OPA and comparison with SRRH-LPO}
		\label{AppendixComplexity}
		SRRH-OPA consists in successively applying SRRH-LPO to set the subcarrier-RRH assignment, and afterwards applying the optimal PA described in \cite{R14}. Therefore, the complexity of SRRH-OPA equals that of SSRH-LPO added to the complexity of optimal PA which is discussed next.\\
		
		Following the optimal power formulation provided in \cite{R14}, the relaxed version of the problem is as follows:\\
		Let $K_n$ be the set of multiplexed users on subcarrier $n$, $\mathcal{N}_M$ the set of multiplexed subcarriers, $\mathcal{S}^{sole}$ the set of sole subcarriers with $N^{sole} = \vert \mathcal{S}^{sole} \vert$, $k_1(n)$ the first user over the subcarrier $n$, where $n$ is either a sole or a multiplexed subcarrier, $k_2(n)$ the second user over the subcarrier $n$, where $n$ is a multiplexed subcarrier, $r(n)$ the RRH powering the signals on the subcarrier $n$, $R_{k,n,r}$ the rate achieved by user $k$ on subcarrier $n$ powered by RRH $r$. Using the same rate to power conversion procedure as in \cite{R14}, the optimization problem can be expressed as follows:
		\begin{multline*}
		\min_{R_{k,n,r}}\sum_{n\in \mathcal{N}_M \cup \mathcal{S}^{sole}}\!\!\!\! \frac{a(n)\sigma^2}{h_1(n)}+\frac{(b(n)-1)\sigma^2}{h_2(n)}\left[\frac{1}{h_2(n)}+\frac{a(n)-1}{h_1(n)}\right]
		\end{multline*}
		Subject to: $$ \mathleft \sum_{n\in S_k} R_{k,n,r(n)}=R_{k,req},\forall k\in 1:K $$
		Where $h_1(n)=h_{k_1(n),n,r(n)}^2$, $h_2(n)=h_{k_2(n),n,r(n)}^2$, $\sigma^2 = N_0B/S$, $a(n)=2^{R_{k_1(n),n,r(n)}S/B}$, and $b(n)=2^{R_{k_2(n),n,r(n)}S/B}$. $R_{k_1(n),n,r(n)}$ is the rate achieved by the strong or sole user $k_1(n)$ on subcarrier $n$, and $R_{k_2(n),n,r(n)}$ is the rate delivered on the subcarrier $n$ to the user $k_2 (n)$. If $n$ happens to be a sole subcarrier, than $R_{k_2(n),n,r(n)}$ is null. The Lagrangian of this problem is given by:
		\begin{multline*}		
		L(R_{k,n,r},\lambda) = 
		\left[ \frac{(a(n)-1)\sigma^2}{h_1(n)} + \frac{\sigma^2}{h_2(n)}\right]		\frac{b(n)-1}{h_2(n)}\\
				+\sum_{n\in \mathcal{N}_M \cup \mathcal{S}^{sole}}\w\!\!\!\left(  a(n)-1  \right)\frac{\sigma^2}{h_1(n)} - \sum_{k=1}^{K}\lambda_k\left(\sum_{n=1}^{N}R_{k,n,r(n)}-R_{k,req}\!\right)   
		\end{multline*}
		After applying the KKT conditions, and including the $K$ rate constraints, we obtain a system of $ N _{sole} +2card(\mathcal{N} _M )+K$ non-linear equations and unknowns ($ N _{sole} +2card(\mathcal{N} _M)$ rate variables and $K$ Lagrangian multipliers). A numerical solver is used to determine the solution, namely the trust-region dogleg method. Since finding an exact expression of this method's complexity is cumbersome, we propose to provide instead the average execution time ratio of 
		SRRH-OPA with respect to SRRH-LPO, measured over a total of 1000 simulations at a rate of 12Mbps, for $K=15$ users, $S=64$ subcarriers and $R=4$ RRHs. We observed that the execution time of SRRH-OPA is more than the double the one of SRRH-LPO, while the performance improvement is of only 2\%. This showcases the efficiency of our LPO procedure, both in terms of its global optimal-like performance and in terms of its cost effective implementation.
	\end{appendices}

	\bibliographystyle{IEEEtran}


	\begin{IEEEbiography}
		[{\includegraphics[width=1in,height=1.25in,clip,keepaspectratio]{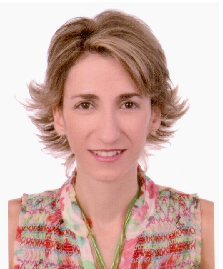}}]
		{Joumana Farah}
		Joumana Farah received the B.E. degree in Electrical Engineering from the Lebanese University, in 1998, the M.E. degree in Signal, Image, and Speech processing, in 1999, and the Ph.D. degree in mobile communication systems, in 2002, from the University of Grenoble, France. Since January 2010, she has held a Habilitation to Direct Research (HDR) from University Pierre and Marie Curie (Paris VI), France. She is currently a full-time professor with the Faculty of Engineering, Lebanese University, Lebanon. She has supervised a large number of Master and PhD theses. She has received several research grants from the Lebanese National Council for Scientific Research, the Franco-Lebanese CEDRE program, and the Lebanese University. She has nine registered patents and software and has coauthored a research book and more than 90 papers in international journals and conferences. Her current research interests include resource allocation techniques, channel coding, channel estimation, interference management, and distributed video coding. She was the General Chair of the 19th International Conference on Telecommunications (ICT 2012), and serves as a TPC member and a reviewer for several journals and conferences. 
	\end{IEEEbiography}

	\begin{IEEEbiography}
	[{\includegraphics[width=1in,height=1.25in,clip,keepaspectratio]{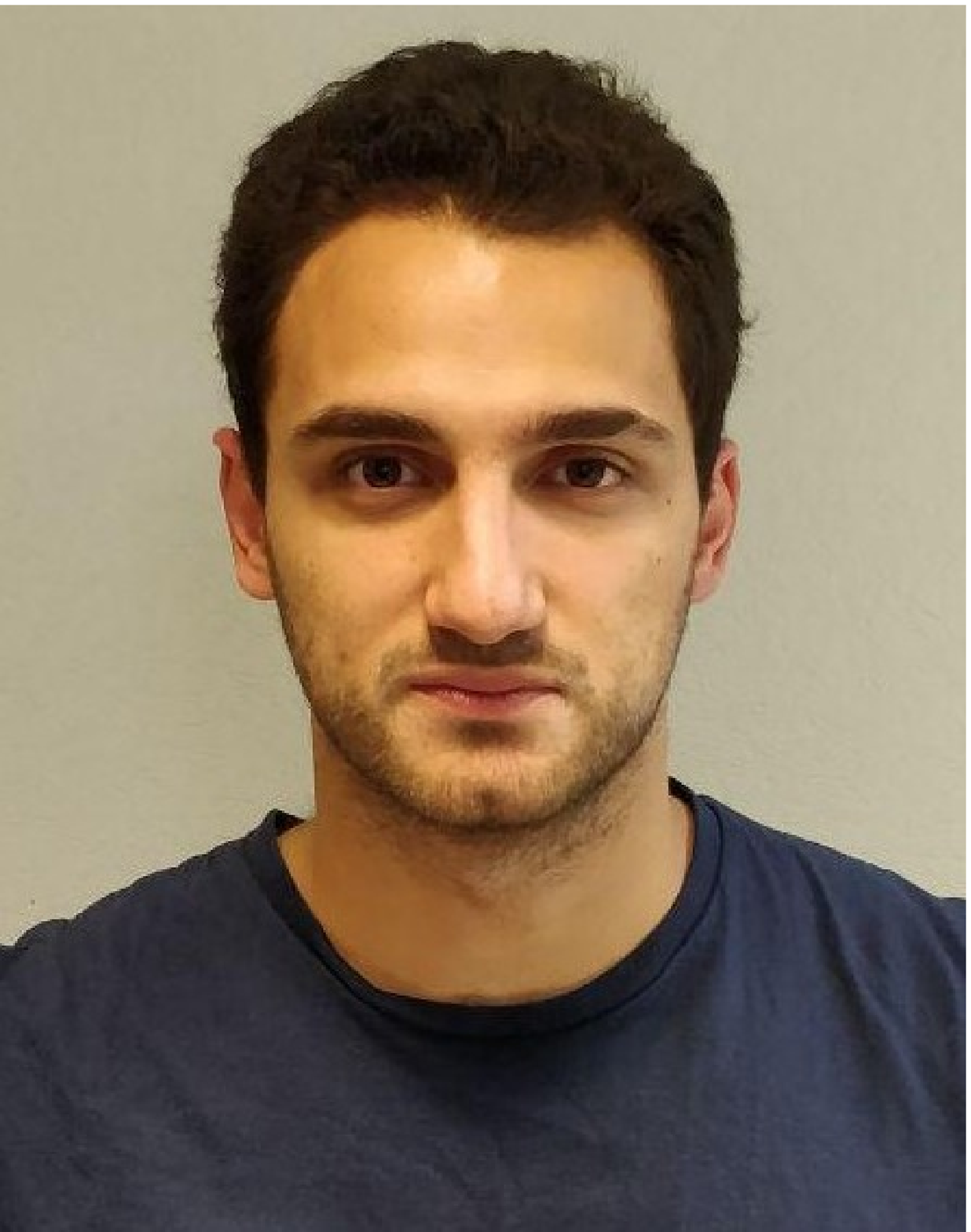}}]
	{Antoine Kilzi}
	Antoine Kilzi received his computer and communications engineering degree in 2017 from the Lebanese University. He is currently working towards the Ph.D. degree in  information and communication engineering at IMT Atlantique. His current research interests include resource allocation, non-orthogonal multiple access and coordinated multipoint systems.
\end{IEEEbiography}

		\begin{IEEEbiography}
		[{\includegraphics[width=1in,height=1.25in,clip,keepaspectratio]{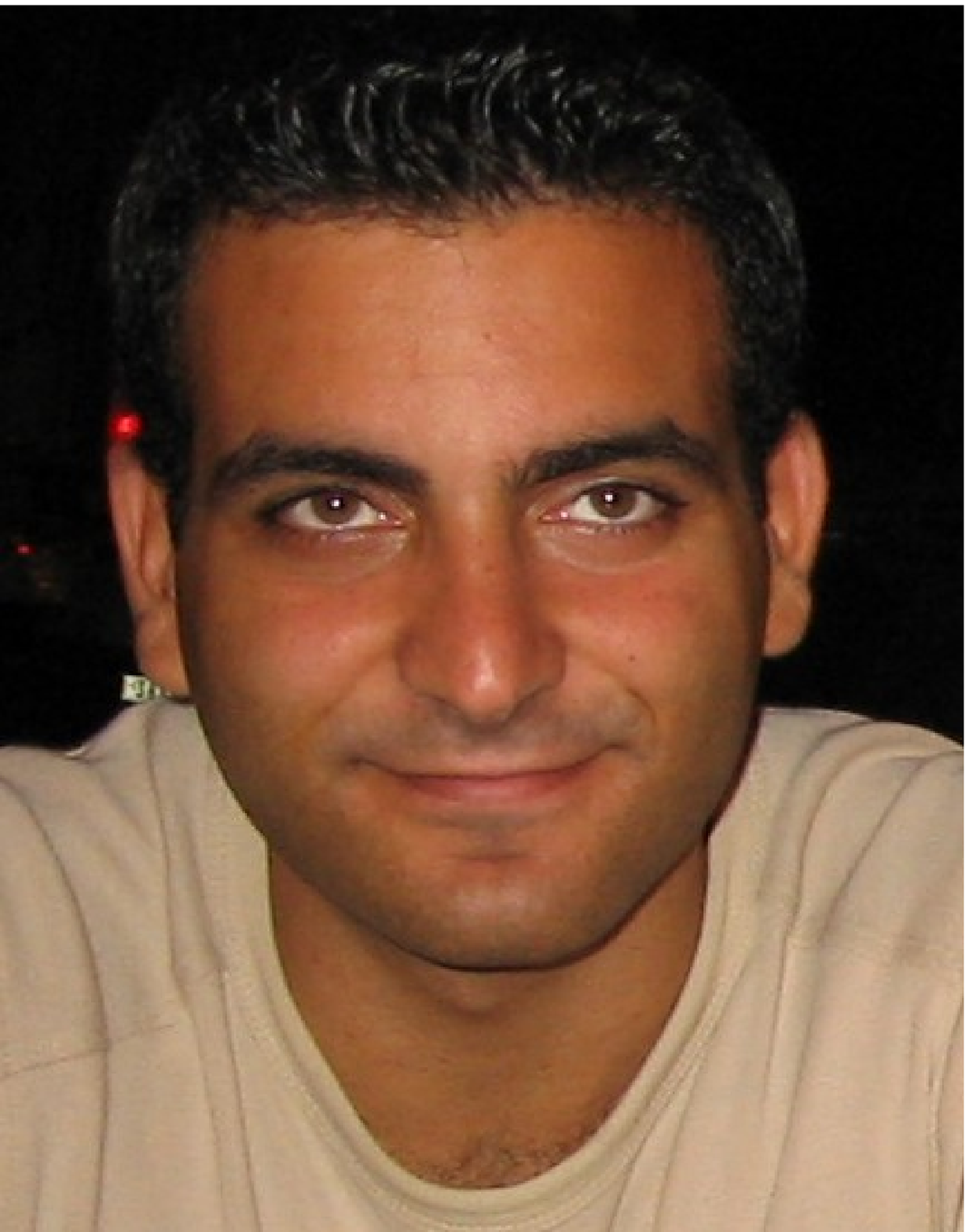}}]
		{Charbel Abdel Nour}
		Charbel Abdel Nour obtained his computer and communications engineering degree in 2002 from the Lebanese University, his Masters degree in digital communications from the University of Valenciennes, France, in 2003 and his PhD in digital communications from Telecom Bretagne, France in 2008. From June 2007 till October 2011, he worked as a post-doctoral fellow at the Electronics Department of Telecom Bretagne. He was involved in several research projects related to broadcasting and satellite communications. Additionally during the same period, he was active in the Digital Video Broadcasting DVB consortium where he had important contributions. Starting November 2011, Charbel holds an associate professor position at the Electronics Department of Telecom Bretagne. His interests concern the radio mobile communications systems, broadcasting systems, coded modulations, error correcting codes, resource and power allocation for NOMA, waveform design, MIMO and iterative receivers. Lately, he presented several contributions to the H2020 METIS and FANTASTIC5G projects and to the 3GPP consortium related to coding solutions for 5G.
	
	\end{IEEEbiography}

\begin{IEEEbiography}
	[{\includegraphics[width=1in,height=1.25in,clip,keepaspectratio]{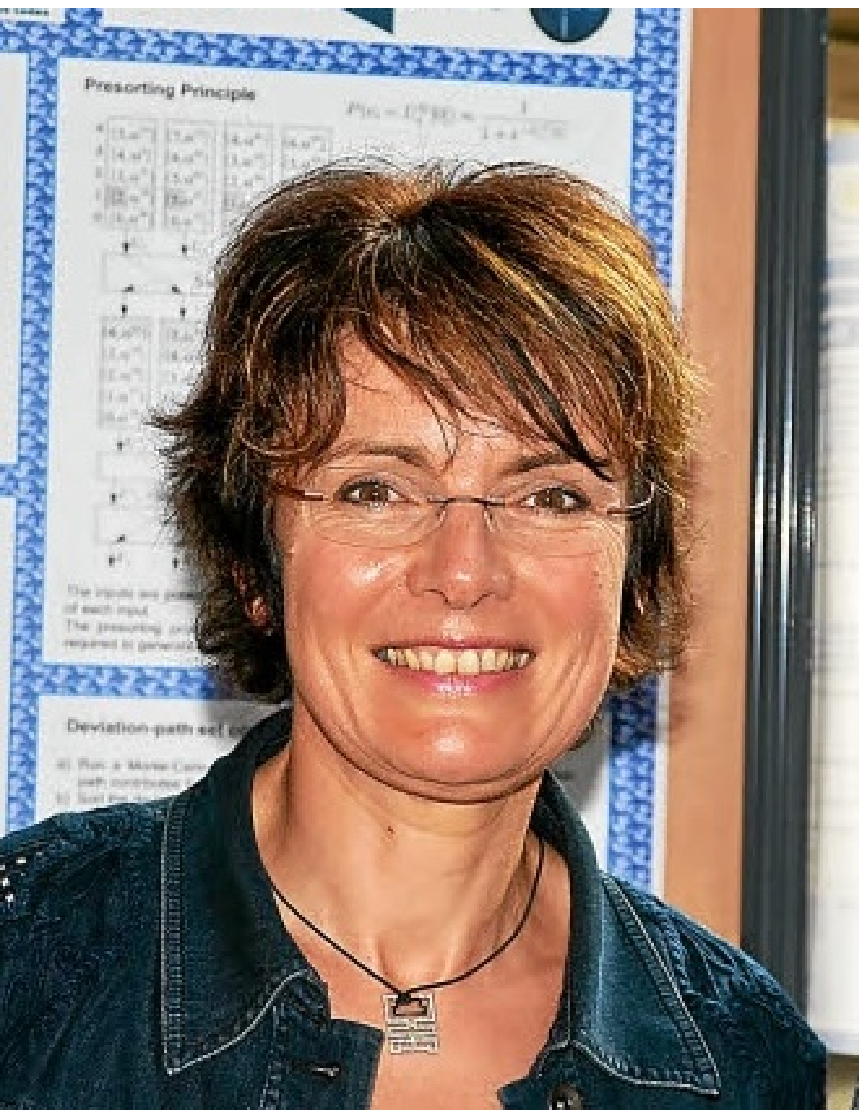}}]
	{Catherine Douillard}
	Catherine Douillard received the engineering degree in telecommunications from the Ecole Nationale Supérieure des Télécommunications de Bretagne, France, in 1988, the Ph.D. degree in electrical engineering from the University of Western Brittany, Brest, France, in 1992, and the accreditation to supervise research from the University of Southern Brittany, Lorient, France, in 2004. She is currently a full Professor in the Electronics Department of Telecom Bretagne where she is in charge of the Algorithm-Silicon Interaction research team. Her main research interests are turbo codes and iterative decoding, iterative detection, the efficient combination of high spectral efficiency modulation and turbo coding schemes, diversity techniques and turbo processing for multi-carrier, multi-antenna and multiple access transmission systems. In 2009, she received the SEE/IEEE Glavieux Award for her contribution to standards and related industrial impact. She was active in the DVB (Digital Video Broadcasting) Technical Modules for the definition of DVB-T2, DVB-NGH as chairperson of the "Coding, Constellations and Interleaving" task force and DVB-RCS NG standards. Since 2015, she has had several contributions in the FANTASTIC-5G and EPIC H2020 European projects intended for the definition of new techniques for 5G and beyond. 
\end{IEEEbiography}
\end{document}